\documentclass[11pt]{article}

\usepackage{fullpage}

\usepackage{dsfont}
\usepackage{xspace}
\usepackage{tikz}
\usetikzlibrary{patterns}
\usepackage{latexsym}
\usepackage{natbib}
\usepackage{amsmath}
\usepackage{amsthm}
\usepackage{amssymb}
\usepackage{amsfonts}
\usepackage{mathrsfs}
\usepackage{aliascnt}
\usepackage{pstricks}
\usepackage{pst-all}
\usepackage{pstricks-add}
\usepackage{pst-plot}
\usepackage{graphicx}
\usepackage{subfig}
\usepackage{float}
\usepackage[multiple]{footmisc}

\usepackage{enumerate}

\usepackage[pdfpagelabels,pdfpagemode=None]{hyperref}
\usepackage{cleveref}
\usepackage{algorithmic}
\usepackage{algorithm}

\newcommand{\STOC}[1]{}
\newcommand{\costasnote}[1]{{\red #1}}

\newtheorem{theorem}{Theorem}%

\newaliascnt{lemma}{theorem}
\newtheorem{lemma}[lemma]{Lemma}%
\aliascntresetthe{lemma}

\newaliascnt{claim}{theorem}
\aliascntresetthe{claim}

\newaliascnt{corollary}{theorem}
\newtheorem{corollary}[corollary]{Corollary}%
\aliascntresetthe{corollary}

\newaliascnt{proposition}{theorem}
\aliascntresetthe{proposition}

\newaliascnt{remark}{theorem}

\aliascntresetthe{remark}

\newaliascnt{algo}{procedure}
\aliascntresetthe{algo}

\theoremstyle{definition}
\newtheorem{definition}{Definition}

\newtheorem{example}{Example}

\makeatletter
\newtheorem*{rep@theorem}{\rep@title}
\newcommand{\newreptheorem}[2]{%
\newenvironment{rep#1}[1]{%
 \def\rep@title{#2 \ref{##1}}%
 \begin{rep@theorem}}%
 {\end{rep@theorem}}}
\makeatother
\newreptheorem{theorem}{Theorem}
\newreptheorem{lemma}{Lemma}

\newcommand{\aref}[1]{\hyperref[#1]{Appendix~\ref{#1}}}

\newcommand{\AutoAdjust}[3]{\mathchoice{ \left #1 #2  \right #3}{#1 #2 #3}{#1 #2 #3}{#1 #2 #3} }
\newcommand{\Xcomment}[1]{{}}

\newcommand{\InBrackets}[1]{\AutoAdjust{[}{#1}{]}}
\newcommand{\Ex}[2][]{\operatorname{\mathbf E}_{#1}\InBrackets{#2}}

\def\expect{\Ex}

\newcommand{\softproblem}{optimal revenue-utility tradeoff\xspace}
\newcommand{\tdsoftproblem}{type-dependent optimal revenue-utility tradeoff\xspace}
\newcommand{\tisoftproblem}{type-independent optimal revenue-utility tradeoff\xspace}

\newcommand{\numitems}{k}

\newcommand{\alloc}{\boldsymbol{x}}

\newcommand{\reals}{\mathbb R}
\newcommand{\agent}{\kappa}

\newcommand{\ironed}{\hat{\boldsymbol{\phi}}}

\newcommand{\constrained}{\hat}
\newcommand{\optconstrained}{\composed{\optimized}{\constrained}}
\newcommand{\optimized}{\starred}
\newcommand{\differentiated}[1]{#1'}
\newcommand{\fortype}{\tilde}

\newcommand{\starred}[1]{#1^\star}
\newcommand{\noaccents}[1]{#1}
\newcommand{\composed}[3]{#1{#2{#3}}}

\newcommand{\forexquant}[1]{#1^{\exquant}}

\newcommand{\newagentvar}[3][\noaccents]{%
\expandafter\newcommand\expandafter{\csname #2\endcsname}{#1{#3}}%
\expandafter\newcommand\expandafter{\csname #2s\endcsname}{#1{\boldsymbol{#3}}}%
\expandafter\newcommand\expandafter{\csname #2smi\endcsname}[1][i]{#1{\boldsymbol{#3}}_{-##1}}%
\expandafter\newcommand\expandafter{\csname #2i\endcsname}[1][i]{#1{#3}_{##1}}%
\expandafter\newcommand\expandafter{\csname #2ith\endcsname}[1][i]{#1{#3}_{(##1)}}%
}

\newagentvar[\tilde]{talloc}{\alloc}

\newagentvar{quant}{q}
\newagentvar[\constrained]{exquant}{\quant}
\newagentvar[\constrained]{critquant}{\quant}  
\newagentvar[\optconstrained]{monoq}{\quant}

\newagentvar{qprice}{\price}
\newagentvar{qrev}{R}
\newagentvar[\ironed]{iqrev}{\qrev}

\newagentvar{qalloc}{\alloc}
\newagentvar{cumalloc}{X}
\newagentvar[\constrained]{calloc}{\qalloc}
\newagentvar[\composed{\forexquant}{\constrained}]{callocstep}{\qalloc}
\newagentvar[\forexquant]{qallocstep}{\qalloc}
\newagentvar[\constrained]{cumcalloc}{\cumalloc}
\newagentvar[\ironed]{ialloc}{\qalloc}
\newagentvar[\ironed]{icumalloc}{\cumalloc}

\newagentvar{outcome}{w}
\newagentvar[\fortype]{toutcome}{\outcome}
\newagentvar[\composed{\forexquant}{\fortype}]{toutcomestep}{\outcome}

\newagentvar{rev}{R}
\newagentvar[\differentiated]{marg}{\rev}
\newagentvar{rawrev}{P}
\newagentvar[\differentiated]{rawmarg}{\rawrev}

\newagentvar[\tilde]{pseudorev}{\rev}
\newagentvar[\tilde]{pseudorawrev}{\rawrev}

\newagentvar{val}{v}

\newagentvar{Val}{V}

\newagentvar{price}{p}

\newagentvar[\constrained]{critval}{\val}
\newagentvar[\constrained]{reserve}{\val} 
\newagentvar[\optconstrained]{monop}{v}

\newagentvar{strat}{s}
\newagentvar{bid}{b}

\newagentvar{dist}{F}

\newagentvar[\ironed]{iprice}{\price}  
\newagentvar[\ironed]{ival}{\val}  

\newagentvar{ints}{{\cal I}}

\let\paragraphwithoutperiod\paragraph
\renewcommand{\paragraph}[1]{\paragraphwithoutperiod{#1.}}

\title{Sequential Mechanisms with ex-post Participation Guarantees\footnote{This work was supported by ONR grant N00014-12-1-0999, and NSF Awards CCF-0953960 (CAREER), CCF-1551875 and SES-1254768. Part of this work was done while the authors were visiting the Simons Institute for Theory of Computing.}}

\author{Itai Ashlagi \\
Stanford MS\&E 
\and Constantinos Daskalakis \\
MIT CSAIL \\
\and Nima Haghpanah \\ MIT CSAIL 
}

\begin{document}

\begin{titlepage}
\maketitle

\begin{abstract}
We provide a characterization of revenue-optimal dynamic mechanisms in settings where a monopolist sells $k$ items over $k$ periods to a buyer who realizes his value for item $i$ in the beginning of period $i$. We require that the mechanism satisfies a strong individual rationality constraint, requiring that the \emph{stage utility} of each agent be positive during each period.  We show that the optimum mechanism can be computed by solving a nested sequence of static (single-period) mechanisms that optimize a tradeoff between the surplus of the allocation and the buyer's utility. We also provide a simple dynamic mechanism that obtains at least half of the optimal revenue. The mechanism either ignores history and posts the optimal monopoly price in each period, or allocates with a probability that is independent of the current report of the agent and is based only on previous reports. Our characterization extends to multi-agent auctions.  We also formulate a discounted infinite horizon version of the problem, where we study the performance of ``Markov mechanisms.''
\end{abstract}

\thispagestyle{empty}
\end{titlepage}

\newpage

\section{Introduction}
%


How should a monopolist sell an item to a buyer whose value for the item will only be realized next week? For example, consider selling a flight to some executive who may or may not have a meeting with a client next week.  Suppose that both the seller and the buyer only know a distribution, $F$, from which the buyer's value, $v$, for the item will be drawn.  One way the seller could go about this is to make a take-it-or-leave-it offer today. The offer reads ``pay $\expect{v}$ today to get the item next week.'' A risk-neutral buyer would find this offer attractive, hence the seller would extract the full surplus, $\expect{v}$. 

The unsettling feature of the afore-described mechanism is that, for some realizations of $v$, the buyer ends up with negative utility. In particular, while  our mechanism is interim Individually Rational (IR), it is not ex-post IR. How could we fix this? One way is to wait until next week, and make a take-it-or-leave-it offer of the item at an optimal monopoly price, i.e. some price $p$ maximizing ${p \cdot (1-F(p))}$. The new mechanism is clearly ex-post IR, but its revenue could be much smaller than that of our previous mechanism. Quite naturally, our new mechanism extracts the best possible revenue among all ex-post IR mechanisms, as a simple argument can establish.  One practical reason to study optimal mechanisms subject to ex-post IR conditions is consumer protection laws~\cite{KrS15}.  For example, the European Union adopted a legislation in 2011 demanding online retailers to give buyers the right for free return, effectively ensuring ex-post IR, since the buyer may not know the value for an item bought online before inspecting it.

Now let us consider a slightly more complex scenario, where our executive is a frequent flyer who may be attending meetings every week depending on client needs. Every week $i$ his value, $v_i$, for flying that week  is drawn from a known distribution, $F_i$. How should a seller sell tickets to such an executive? To build intuition let us consider the case of two weeks. Suppose that our executive has already realized his value for week $1$ and is as uncertain as the seller about his value for week $2$. What is the best way to sell to such a buyer? Extending our interim IR mechanism from before, we can sell both flights today, offering this week's flight  for the optimal monopoly price under $F_1$ and next week's flight for $\expect{v_2}$. Again, this mechanism is interim IR and optimizes our revenue (by extracting optimal surplus tomorrow and optimal revenue today). On the other hand, the mechanism is not ex-post IR.

It seems natural then that, if we were to insist on satisfying ex-post IR, our only option would be to wait until next week to sell next week's flight at the optimal monopoly price for distribution $F_2$, thereby extracting the sum of today's and next week's optimal revenue. Quite surprisingly this is not the case! Here is an example by Papadimitriou et al.~\citeyear{DBLP:journals/corr/PapadimitriouPPR14} that extracts more revenue than the sum of single day optimal revenues:

\begin{example} \label{example:two equal revenues} Suppose that the value of the first item is drawn from an equal revenue distribution truncated at $n$ and the value of the second item is drawn from an equal revenue distribution truncated at $e^n$, for some constant $n$.\footnote{Recall that the equal revenue distribution has support $[1,+\infty)$, density function $f(x)={1 \over x^2}$ and cumulative density function $F(x)=1-{1 \over x}$. The equal revenue distribution truncated at some threshold $T$ has support $[1, T]$ and density $f_T$ that equals $f$ for $x<T$ and has an atom at $x=T$ of total probability mass $1\over T$.} If we were to run two monopoly pricing mechanisms in sequence, our expected revenue would equal $2$. However, the following auction performs much better. The buyer is requested to submit a bid $b$ in the first stage of the mechanism, and is given the first item at a price of $b$, together with a contract that he will receive the second item at price $0$ and with probability $b \over n+1$. It is easy to check that truthful reporting is a weakly dominant strategy and the proposed mechanism is strongly ex-post IR and has $\Theta(\log n)$ expected revenue. As $n$ is arbitrary, this means that there is an unbounded gap between running two Myerson auctions in sequence and the optimal ex-post IR dynamic mechanism.
\end{example}

\paragraph{Results} In this paper, we provide a characterization of the revenue-optimal, ex-post IR, dynamic mechanism over $k$ days and involving $m$ bidders whose values are independent. In particular, we optimize the seller's revenue subject to the following strong individual rationality condition: at each period, the \emph{stage utility} of each agent, defined to be his surplus from that period's allocation minus the agent's payment, must be non-negative.  In particular, the non-negativity of the stage utilities implies that, at the end of each period, each agent's realized utility from participating in the mechanism so far is non-negative. See~\autoref{thm:characterization} for the single-bidder and \autoref{thm:characterization multi agent} for the multi-bidder characterization results. As an application of our characterization we can argue, e.g., that the mechanism described in Example~\ref{example:two equal revenues} is optimal (see \autoref{ex:equal revenue day 1 contd} for a generalized example). 

Our characterization reveals structural properties of optimal mechanisms.  We show that there exists an optimal mechanism in which, in all periods except for possibly the last, the stage utility of all realized types of the agent is zero; that is, every type is asked to pay its surplus from the allocation. 
Moreover, the mechanism makes simple updates to a scalar state variable that dictates its future allocations and payments.  More precisely, we show that an optimal mechanism can be described via two functions that depend only on the current state of the mechanism and the agent's bid (rather than the full history of bids): an \emph{allocation function} that specifies the probability of allocating the item (and hence the payment due to the afore-described surplus extraction), and a \emph{state update}  function that specifies how the state variable should be updated.  We provide a characterization of allocation and state update functions that will result in feasible mechanisms, and a recursive family of static single-dimensional problems, the solution to which are the optimal allocation and state update functions. These can be identified via backwards induction.  We provide a Fully Polynomial Time Approximation Scheme to compute the description of the optimal mechanism to within any desired error.

Our characterization also allows us to design simple single-bidder mechanisms that guarantee at least half of the optimal revenue for any $k$.  While the optimal mechanism needs to carefully balance revenue gain at each period with updating the state variable in a manner that allows for more revenue in the future, the 2-approximately optimal mechanism is based on much simpler tradeoffs.  We show that randomizing over two simple mechanisms gives a 2-approximation to the optimal revenue.  The first mechanism simply ignores all history and  myopically maximizes revenue in each period.  The second mechanism ignores the bidder's report in each period to compute the allocation in that period. Instead it allocates the item with a probability that only depends on the state variable, which itself is updated as a simple function of each period's allocation probability and bidder report.  
Thus, compared to the optimal mechanism, the 2-approximation is described using fewer parameters. Similarly to the optimal mechanism, the 2-approximation is found via backwards induction.\footnote{We thank Song Zuo for pointing out an issue with the mechanism in an earlier version of the paper, which is corrected in this version.}

In Section~\ref{sec:infinite horizon}, we formulate an infinite horizon version of the problem with discounts, and argue that restricting attention to ``Markov mechanisms,'' whose allocation in each period is homogeneous and only depends on the current and the previous period's report by the bidder, does not improve revenue compared to posting optimal monopoly prices in every period. 

\paragraph{Our Approach} One approach to finding the optimal dynamic mechanism is to attempt a $k$ period dynamic programming formulation. Let us focus on the single-bidder case.  We are seeking an optimal collection of allocation and price rules, $(x_i(v_{\le i}),p_i(v_{\le i}))_{i=1}^k$, where $x_i(v_{\le i})$ represents the probability that the item is allocated to the bidder in period $i$, as a function of his reports $v_{\le i}$ in all periods up to period $i$, and $p_i(v_{\le i})$ records the expected price paid by the bidder in period $i$. What makes the problem challenging is that the choices we make for the allocation and payment in period $i$ will affect the incentive constraints for all periods $<i$. This makes the representation complexity of the internal states of the naive dynamic programming formulation explode.  Yet, our characterization shows that a more tractable dynamic programming formulation exists. The interesting feature of our formulation is that is nests its subproblems in the opposite way than the naive one, namely the last period's optimization sits inside the nested sequence of problems while the first period's optimization sits outside. More importantly, it  maintains a sparse representation of the decisions made by the dynamic program in periods $>i$ that it passes on to period $i$.   These are the \emph{cumulative tradeoff} functions $\hat{g}_i(\cdot)$ in Theorems~\ref{thm:characterization} and~\ref{thm:characterization multi agent}. In particular, the information does not explode (in the number of periods) as we exit the nested optimizations of our dynamic programming formulation. Ultimately, to implement the optimal dynamic mechanism we need to do two passes over the periods, one starting from period $k$ and moving backwards toward period $1$ to find the $\hat{g}_i$'s (the ``preprocessing step" in Theorems~\ref{thm:characterization} and~\ref{thm:characterization multi agent}) and another, taking place as the mechanism interacts with the agent, starting from period $1$ and moving forward towards period $k$ to find the allocation and price rule in each period (Step 1 in the theorems).

At the heart of our characterization result/dynamic programming formulation lies a type of surplus-utility tradeoff problem (see~\autoref{def:revenue utility tradeoff general} and~\autoref{def:revenue utility tradeoff general multi} for our single- and multi-agent characterizations respectively).  The goal in this problem is to optimize a linear combination of the allocation's surplus and a function $g(\cdot)$ of the bidders' utilities, subject to a given constraint on the expected utility of the mechanism. The preprocessing step of our characterization theorems requires solving a sequence of such problems starting from period $k$ and moving backwards towards period $1$. In the absence of the utility term in the objective, this problem can be formulated in terms of the allocation function, and can be point-wise optimized leading to $0$-$1$ allocation rules. With the utility term, it becomes more natural to formulate the problem in terms of the bidder's utility function. In \autoref{sec:utility constrained surplus} we characterize the optimum of this problem and show that the optimal mechanism may involve fractional allocations.

\smallskip To provide some intuition about our characterization results (Theorems~\ref{thm:characterization} and~\ref{thm:characterization multi agent}), let us consider the two-period single bidder case. Our reduction to surplus-utility tradeoff works roughly as follows. The optimal dynamic mechanism reduces to finding an allocation and price rule $(x_1(v_1), p_1(v_1))$ for the first stage, which can only depend on the buyer's value for the first item, along with an allocation and price rule $(x_2(v_1,v_2), p_2(v_1,v_2))$ for the second stage, which may depend on both values. The goal is to maximize revenue (Expression~\eqref{eq:revenue}) subject to IC (Inequalities~
\eqref{eq:incentive compatible}) and IR (Inequalities~\eqref{eq:expostIR}) constraints. The challenge is that the optimizations of $(x_1,p_1)$ and $(x_2,p_2)$ are entangled. In particular, the IC condition for period $1$ involves both $(x_1,p_1)$ and $(x_2,p_2)$. 

We are hence looking for a way to disentangle the optimization problems in the two periods. We use a simple change of variables to rewrite our problem as optimizing the expected (w.r.t. $v_1$) sum of an (adjusted) payment $\hat{p}_1(v_1)$ from the first stage and the buyer's expected surplus $E_{v_2}[v_2 x_2(v_1,v_2)]$ in the second stage. This optimization is subject to the adjusted mechanism $((x_1, \hat{p}_1), (x_2,p_2))$ satisfying incentive compatibility along with the additional constraint that, point-wise w.r.t. $v_1$, the adjusted utility $\hat{u}_1(v_1):=v_1 x_1(v_1)-\hat{p}_1(v_1)$ in the first stage upper bounds the expected (w.r.t. $v_2$) utility of the second stage, given $v_1$. See formulation~\eqref{eq:alternative formulation}. Crucially, the IC constraints on $(x_1,\hat{p}_1)$ do not involve $(x_2,p_2)$ and vice versa. The two problems now only interface through the bound on the utility of $(x_2,p_2)$ as determined by $(x_1,\hat{p}_1)$ and the reported value $v_1$. To capture this interface, we define a \emph{cumulative tradeoff function} $\hat{g}_2$, mapping a given upper bound on the utility of the second stage mechanism to the maximum welfare achievable in the second stage. (See Definitions~\ref{def:cum tradeoff} and~\ref{def:cum tradeoff multi} for general $k$.) Hence, the dynamic mechanism design problem reduces to an instance of surplus-utility tradeoff: optimize the expected (w.r.t. $v_1$) sum of the adjusted payment $\hat{p}_1(v_1)$ in the first stage  and $\hat{g}_2(\hat{u}_1(v_1))$. The latter problem only involves $(x_1,\hat{p}_1)$, while computing the cumulative tradeoff function only involves $(x_2,p_2)$. See Section~\ref{sec:dynamic mechanism design and optimal revenue-utility tradeoff} for more details and generalization to $k$ periods and multiple agents.

\paragraph{Related Work}


The literature on dynamic mechanism design is rather broad (see~\cite{bergemann2011dynamic}), but has a different focus than ours. The main thrusts in this literature study dynamic arrivals and departures of agents, e.g. \cite{parkes2003mdp,pai2008optimal,gershkov2009dynamic,gershkov2010efficient}, or agents whose private information evolves, e.g.~\cite{courty2000sequential, kakade2013optimal, pavan2014dynamic, CavalloPS06, Cavallo08, bergemann2010dynamic, athey2013efficient, esHo2007optimal}. These papers analyze quite general dynamic mechanism design settings, involving several bidders and several stages, but fall short from capturing even our single-bidder two-stage problem. The difference lies in the strong participation constraints that we choose to enforce in this paper, guaranteeing  that all types receive positive utility from having participated in the mechanism so far, at the end of each round. Instead, prior literature considers weaker notions of individual rationality requiring that, in the beginning of each round of the mechanism, the expected utility from all future rounds be positive.  As discussed earlier, the latter notion of individual rationality results in mechanisms that we do not find compelling in our setting, so we are motivated to study the stronger notion of individual rationality. 

Closer to our work are recent works of \cite{DBLP:journals/corr/PapadimitriouPPR14} and \citep{KrS15}, which consider dynamic mechanisms with ex-post IR guarantees.  Papadimitriou et al.~\citeyear{DBLP:journals/corr/PapadimitriouPPR14} consider the same dynamic mechanism design setting that we do, but focus on the computational complexity of finding the optimal dynamic mechanism. They show that when the buyer's values are correlated, finding the optimal deterministic mechanism is {\tt NP}-hard, while the optimal randomized mechanism can be computed via an LP whose size is polynomial in the support of the type distribution. In comparison to that work, we aim at characterizing the structure of the optimal dynamic mechanism and allow randomization.  \cite{KrS15} consider a problem where the seller has a single item to sell and the buyer sequentially receives signals about his valuation.  Their model is thus different from our setting where the seller and the buyer have a common prior from which the values are drawn, and where there are multiple items to sell.  They show that in their setting, the optimal mechanisms are static in the sense that the seller does not elicit the buyer's information sequentially.

Our work is related to the repeated sales and dynamic pricing literature (see~\cite{DevanurPS15,BabaioffDKS12} and references therein). In~\cite{DevanurPS15} and related papers, the seller sells different items to the same buyer over multiple rounds, but is unable to make commitments and must therefore play a Bayesian Nash equilibrium. In~\cite{BabaioffDKS12} and related papers, the seller sells a limited supply of items to a stream of i.i.d. buyers from an unknown distribution, and they use the connection to multi-armed bandits to design competitive mechanisms. Given that the buyers are different in every round there are no incentive constraints across different rounds.

Independently and contemporaneously to our work, \cite{MLT16} consider the same problem studied here, where a seller wishes to maximize revenue subject to an ex-post individual rationality constraint.  They propose a class of mechanisms called \emph{bank account mechanisms}. Bank account mechanisms maintain a scalar state variable, called ``balance,'' that is updated in the course of the execution of the mechanism. In every round, the allocation and the price depend on the bidder's report and the balance, and the update of the balance is specified by a ``spend'' and a ``deposit'' function. Overall a bank account mechanism is described by four functions. In earlier work, \cite{MLT16b} show that bank account mechanisms can be used in dynamic settings to derive optimal mechanisms with interim individual rationality constraints. In \cite{MLT16} they show how to derive optimal mechanisms with ex post individual rationality constraints. The optimal mechanisms identified by our work and theirs both maintain a scalar state variable, but the two mechanisms are different.. For example, our optimal mechanism satisfies a zero stage utility property, i.e. extracts the bidder's full surplus from his allocation in each period except possibly the last. Similar to our simple $2$-approximation, \cite{MLT16} specify a simple bank account mechanism that achieves a $3$-approximation to the optimal revenue.  Finally, \cite{MLT16} provide extensions to multiple items, whereas our work provides extensions to multiple bidders. Our work is straightforwardly extendable to multiple items, and we believe their work is extendable to multiple bidders.


\section{Preliminaries}\label{sec:prelim}
We consider a dynamic mechanism design problem, where a seller sells $\numitems$ items sequentially in $\numitems$ stages. The buyer has value $v_i \in [\underline{v}_i,\bar{v}_i]$ for item $i\in [1:\numitems]$, drawn independently from his other values from a distribution with density $f_i$ and cumulative density $F_i$. Our results also apply to distributions with discrete support, but we will restrict our attention to distributions with a density function. Moreover, whenever convenient we may assume without loss of generality that $\underline{v}_i=0$. Back to our dynamic mechanism setting, we assume that the value for each item is revealed to the buyer in the beginning of the corresponding stage; in particular, the buyer only knows $v_1,\ldots,v_i$, denoted $v_{\leq i}$, when buying item $i$. The goal is to design a revenue optimal mechanism for selling these $k$ items with strong participation guarantees, as formalized below.

We use the revelation principle and design direct incentive compatible mechanisms.  A mechanism is a sequence of allocation probability functions $x_i(v_{\leq i}) \in [0,1]$ and payment functions $p_i(v_{\leq i}) \in \reals$, for $i\in [1:\numitems]$.  A mechanism is \emph{periodic incentive compatible} (PIC) if at any stage $i$, revealing $v_i$ truthfully maximizes the agent's expected utility, given truthfulness in the following stages, that is,
\begin{align}
v_i x_i(v_{\leq i}) &- p_i(v_{\leq i}) + \expect[v_{i+1},\ldots,v_k]{\sum_{j>i} x_j(v_{\leq j})v_j - p_j(v_{\leq i}) } \nonumber\\ \geq & v_i x_i(v_{\leq i},v'_i) - p_i(v_{\leq i},v'_i) + \expect[v_{i+1},\ldots,v_k]{\sum_{j>i} x_j(v_{\leq j},v'_i)v_j - p_j(v_{\leq j},v'_i) }, \label{eq:incentive compatible}
\intertext{for all $v_{\leq i}$ and $v'_i$, where $(v_{\leq j},v'_i)$ is a vector of size $j$ in which the $i$'th index is replaced by $v'_i$. A mechanism is ex-post individually rational if the agent's utility is non-negative at each stage, that is,}
v_i x_i(v_{\leq i}) &- p_i(v_{\leq i}) \geq 0.\text{\footnotemark}\label{eq:expostIR}
\intertext{The goal is to maximize the sum of the payments}
&\expect[v_1,\ldots,v_k]{\sum_i p_i(v_{\leq i})}, \label{eq:revenue}
\end{align}\footnotetext{Note that $x_i$ denotes the probability of allocation.  Thus, the ex-post individual rationality states that the utility of the agent is non-negative in expectation over the randomization of the mechanism.  However, the solution can be converted to one that ensures ex-post individual rationality, even for random choices of the mechanism, by correlating payment with allocation and charging the agent $p_i/x_i$ when the item is allocated, and zero otherwise.  The ex-post individual rationality constraint ensures that if the item is allocated, the payment is less than the value since $x_i v_i - p_i \geq 0$ implies that  $v_i - p_i/x_i \geq 0$.}

\noindent subject to the periodic incentive compatibility and ex-post individual rationality constraints.

\subsection{Standard Analysis for $k=1$}\label{sec:myerson}
The following standard analysis relates allocation, payment, and utility functions, and expresses revenue in terms of the allocation function for the special case where $k=1$.  In this case, a mechanism is simply a pair of allocation function $x(v)$ and payment function $p(v)$, and the incentive compatibility constraint is that 
\begin{align}
vx(v) - p(v) \geq vx(v') - p(v'), \forall v, v'.\label{eq:staticIC}
\end{align}

We will distinguish inequalities \eqref{eq:incentive compatible} and \eqref{eq:staticIC} by referring to the former as periodic incentive compatibility, and the later simply as incentive compatibility (even though incentive compatibility is a special case of periodic incentive compatibility for $k=1$).

\begin{lemma}[\citet{M81,Rochet85}]\label{lem:myerson's lemma}
For $k=1$, a mechanism $(x,p)$ is incentive compatible if and only if the allocation fucntion $x(\cdot)$ is monotone non-decreasing, and the allocation function $x$ and the payment function $p$ satisfy $p(v) = v\cdot x(v) - \int_{z\geq \underline{v}}^v x(z) dz + p(\underline{v})$. The utility function $u(v) = vx(v) - p(v)$ of an incentive compatible mechanism is $u(v) = \int_{z\geq \underline{v}}^v x(z) dz - p(\underline{v})$. 

Alternatively, a mechanism is inventive compatibly if and only if the utility function $u(\cdot)$ is convex and non-decreasing, and is differentiable almost everywhere.  In that case, the allocation and payment functions satisfy $x(v)=u'(v)$ and $p(v) = vu'(v) - u(v)$ wherever the utility function is differentiable.
\end{lemma}

Myerson showed that given the above lemma, the expected revenue of an incentive compatible mechanism can be re-expressed using integration by parts
\begin{align}
E_v[p(v)]  = E_v[v\cdot x(v) - \int_{z\leq v} x(z) dz] + p(\underline{v})  = E_v[x(v) \phi(v)] + p(\underline{v}),\label{eq:virtual value}
\end{align}
where $\phi$ is the \emph{virtual value} function $\phi(v) = v - \frac{1-F(v)}{f(v)}$.

\begin{example}[The Equal Revenue Distribution]\label{ex:equal revenue}
Consider the \emph{equal revenue} distribution with $f(v) = 1/v^2$ over the support $v \in [1,\infty)$.  The virtual value function is $\phi(v) = v - \frac{1/v}{1/v^2} = 0$.  Myerson's analysis implies that the revenue of any IC mechanism is $p(1)$.   Given the IR condition, $p(1) \leq x(1) \leq 1$.  The optimal revenue is therefore 1, which is achieved by posting a price $1$.
\end{example}

The following fact is standard and follows from the above analysis.
\begin{lemma}\label{lem:randomization over prices}
Any incentive compatible mechanism is a distribution over posted prices, and a transfer $p(0)$.
\end{lemma}

\section{Optimal Dynamic Mechanisms} \label{sec:dynamic mechanism design and optimal revenue-utility tradeoff}
One of the main challenges in identifying optimum solutions to dynamic mechanism design problems is the complex structure of the periodic incentive compatibility conditions (in our setting, inequality \eqref{eq:incentive compatible}).  In a truly dynamic mechanism, the allocation and payment decisions depend on the history of agent's decisions, which complicates the periodic incentive compatibility condition since the agent's report at a day $i$ not only affects his allocation and payment at that day, but also all future allocation and payments.  In this section, we reformulate the problem in an alternative form via a simple change of variables.  The main purpose of the alternative formulation is to simplify the periodic incentive compatibility conditions to ones that resemble incentive compatibility in static sense (inequality \eqref{eq:staticIC}) more closely.  The following sections will heavily utilize the alternative formulation to decompose and solve the problem.

\smallskip To simplify notation define the expected utility of the agent from future transactions
\begin{align}
U_i(v_{\leq i}) &:= \expect[v_{i+1},\ldots,v_k]{\sum_{j>i} x_j(v_{\leq j})v_j - p_j(v_{\leq j}) },\label{eq:expectedutility}
\intertext{and rewrite the periodic incentive compatibility condition \eqref{eq:incentive compatible} as}
v_ix_i(v_{\leq i}) - p_i(v_{\leq i}) + U_i(v_{\leq i}) &\geq  v_ix_i(v_{\leq i}, v'_i) - p_i(v_{\leq i}, v'_i) + U_i(v_{\leq i}, v'_i).\nonumber
\intertext{Define the \emph{adjusted payment} function }
\hat{p}_i(v_{\leq i}) &:= p_i(v_{\leq i}) - U_i(v_{\leq i}). \label{eq:adjusted price}
\intertext{We will use the above change of variables to formulate the problem in terms of the allocation functions $x_i$ and the adjusted payment function $\hat{p}_i$. The PIC constraint \eqref{eq:incentive compatible} can be rewritten as}
v_i x_i(v_{\leq i}) - \hat{p}_i(v_{\leq i}) & \geq v_i x_i(v_{< i}, v'_i) - \hat{p}_i(v_{<i},v'_i), \forall v_{\leq i},v_i'. \nonumber
\end{align}

\noindent That is, a mechanism $(x,p)$ is periodic incentive compatible if and only if for each $i$ and $v_{< i}$, the mechanism $(x_i(v_{<i},\cdot),\hat{p}_i(v_{<i},\cdot))$, viewed as a static mechanism of only one variable $v_i$, is incentive compatible.  We will next express the expected revenue in terms of the adjusted payment function.  Notice that the expected utility from future transactions~\eqref{eq:expectedutility} can be written recursively as follows:

\begin{align}
U_i(v_{\leq i}) &= \expect[v_{i+1},\ldots,v_k]{x_{i+1}(v_{\leq i+1})v_{i+1} - p_{i+1}(v_{\leq i+1}) + U_{i+1}(v_{\leq i+1}) }\nonumber\\
&= \expect[v_{i+1},\ldots,v_k]{x_{i+1}(v_{\leq i+1})v_{i+1} - \hat{p}_{i+1}(v_{\leq i+1}) }. \label{eq:recursiveutility}
\intertext{(For the purposes of defining $U_k$ set $x_{k+1}=\hat{p}_{k+1}=0$.)  The definition of the adjusted payment $\hat{p}_i$ \eqref{eq:adjusted price} and equation \eqref{eq:recursiveutility} imply that}
p_i(v_{\leq i}) &= \hat{p}_i(v_{\leq i}) + U_i(v_{\leq i}) \nonumber \\
&= \hat{p}_i(v_{\leq i}) + \expect[v_{i+1},\ldots,v_k]{x_{i+1}(v_{\leq i+1})v_{i+1} - \hat{p}_{i+1}(v_{\leq i+1})}\nonumber \\
&= \hat{p}_i(v_{\leq i}) + \expect[v_{i+1}]{x_{i+1}(v_{\leq i+1})v_{i+1} - \hat{p}_{i+1}(v_{\leq i+1})}.\label{eq:the other conversion}
\intertext{Summing up the above equality for all $i$ gives an alternative expression for revenue}
\expect[v_1,\ldots,v_k]{\sum_i  p_i(v_{\leq i})} &= \expect[v_1,\ldots,v_k]{\sum_i (\hat{p}_i(v_{\leq i}) + x_{i+1}(v_{\leq i+1})v_{i+1} - \hat{p}_{i+1}(v_{\leq i+1}))} \nonumber \\
&= \expect[v_1,\ldots,v_k]{\hat{p}_1(v_1) + \sum_{i=2}^k x_i(v_{\leq i}) v_i}.\nonumber
\end{align}
Finally, the ex-post individual rationality constraints can be written as
\begin{align}
v_i x_i(v_{\leq i}) - \hat{p}_i(v_{\leq i}) &\geq U_i(v_{\leq i}) \nonumber \\
&= \expect[v_{i+1}]{x_{i+1}(v_{\leq i+1})v_{i+1} - \hat{p}_{i+1}(v_{\leq i+1})},\label{eq:utility constraint}
\end{align}

\noindent where the last equality followed by \eqref{eq:recursiveutility}.  Given the above discussion we can reformulate the problem in terms of $x$ and $\hat{p}$ variables.  The discussion is summarized into the following re-formulation of the problem, and the lemma below.  To avoid confusion we will refer to problem \eqref{eq:alternative formulation} below specified in terms of $x$ and $\hat{p}$ as the \emph{adjusted problem}, and the problem of maximizing revenue \eqref{eq:revenue} subject to periodic incentive compatibility \eqref{eq:incentive compatible} and ex-post individual rationality \eqref{eq:expostIR} formulated in terms of $x$ and $p$ as the \emph{original problem}.

\begin{align}
\max_{x,\hat{p}}  &\expect{\hat{p}_1(v_1) + \sum_{i=2}^k x_i(v_{\leq i}) v_i} & \label{eq:alternative formulation} \\
\text{s.t., } & v_i x_i(v_{\leq i}) - \hat{p}_i(v_{\leq i})  \geq v_i x_i(v_{< i}, v'_i) - \hat{p}_i(v_{<i},v'_i) & \forall i, v_{\leq i},v'_i \label{eq:adjustedIC} \\
& v_i x_i(v_{\leq i}) - \hat{p}_i(v_{\leq i}) \geq \expect[v_{i+1}]{x_{i+1}(v_{\leq i+1})v_{i+1} - \hat{p}_{i+1}(v_{\leq i+1})}& \forall i, v_{\leq i}. \label{eq:adjustedIR}
\end{align}

\begin{lemma}\label{lem:changeofvariables}
A mechanism $(x,p)$ is a feasible solution to the original problem if and only if the mechanism $(x,\hat{p})$ defined via Equation~\eqref{eq:adjusted price}, $\hat{p}_i(v_{\leq i}) := p_i(v_{\leq i}) - U_i(v_{\leq i})$ is a feasible solution to the adjusted problem.  Conversely, $(x,p)$ can be obtained from $(x,\hat{p})$ via Equation~\eqref{eq:the other conversion}, $p_i(v_{\leq i}) = \hat{p}_i(v_{\leq i}) + \expect[v_{i+1}]{x_{i+1}(v_{\leq i+1})v_{i+1} - \hat{p}_{i+1}(v_{\leq i+1})}$.   The revenue of $(x,p)$ in the original formulation is equal to the objective value of $(x,\hat{p})$ in the adjusted formulation~\eqref{eq:alternative formulation}.  In particular $(x,p)$ is an optimal solution to the original problem if and only if $(x,\hat{p})$ is an optimal solution to the adjusted problem.
\end{lemma}

The above lemma decomposes the problem into the design of a sequence of mechanisms where each mechanism $x_i,\hat{p}_i$ is incentive compatible \emph{in isolation} (i.e., as demanded in \eqref{eq:adjustedIC}) for each history of bids $v_{<i}$, and the utilities of the mechanisms are constrained by \eqref{eq:adjustedIR}.  In \autoref{subsec:char} and \autoref{subsec:2approx} we will use this decomposed formulation to characterize the optimum mechanisms and identify approximately optimal solutions.

\subsection{Characterization of the Optimal Single-Bidder Dynamic Mechanism}\label{subsec:char}
We start this section by observing a structural property that the optimal mechanism satisfies, which will simplify its form.  In particular, we observe that there exists an optimal solution to the adjusted problem \eqref{eq:alternative formulation} that satisfies all the utility bounds \eqref{eq:adjustedIR} {for $i<k$} with equality.  To show this, we argue that given a solution to the adjusted problem where some utility bounds are not tight, we can construct another solution where those utility bounds are tight and the objective value remains unchanged.  In particular, consider a feasible solution $(x,\hat{p})$ to the adjusted problem where for some $i$, $v_{\leq i}$, and $\delta>0$,

\begin{align*}
v_i x_i(v_{\leq i}) - \hat{p}_i(v_{\leq i}) &= \delta + \expect[v_{i+1}]{x_{i+1}(v_{\leq i+1})v_{i+1} - \hat{p}_{i+1}(v_{\leq i+1})},
\intertext{Now consider another solution $(y,\hat{q})$ that agrees everywhere with $(x,\hat{p})$, except that for all $j>i$ and $v_{\leq j}$ that contain $v_{\leq i}$ as a prefix, that is, $v_{\leq j} = (v_{\leq i}, v_{i+1},\ldots, v_j)$,}
\hat{q}_j(v_{\leq j}) &= \hat{p}_j(v_{\leq j}) - \delta.
\end{align*}
\noindent Note that as a result of this change, all the IC constraints \eqref{eq:adjustedIC} remain satisfied since each inequality either remains unchanged or $\delta$ is added to both sides.  Similarly, in all IR constraints~\eqref{eq:adjustedIR} either $\delta$ is added to both sides or neither side, except for $i$ and $v_{\leq i}$, which now satisfies the constraint with equality.  Note that since the change only applied to the payment functions, and it did not apply to day 1, the objective value remains unchanged.  By applying the same argument to all non-binding constraints~\eqref{eq:adjustedIR}, we conclude that there exists an optimal solution where all IR constraints~\eqref{eq:adjustedIR} bind.  We combine this observation with \autoref{lem:changeofvariables} to obtain the following lemma.

\begin{lemma}\label{lem:martingalemechanism}
There exists an optimal solution $(x,p)$ to the original problem where is $p_i(v_{\leq i}) = v_i x_i(v_{\leq i})$ for all $i < k$.  That is, the stage utility of all types of the agent is zero in all days before the last day.
\end{lemma}
\begin{proof}
Consider an optimal solution $(x,\hat{p})$ to the adjusted problem where additionally all IR constraints~\eqref{eq:adjustedIR} bind (such a solution exists as argued above).  By~\autoref{lem:changeofvariables}, the mechanism $(x,p)$ obtained via transformation~\eqref{eq:adjusted price} is an optimal solution to the original problem.  In particular,
\begin{align*}
p_i(v_{\leq i}) &= \hat{p}_i(v_{\leq i}) + \expect{x_{i+1}(v_{\leq i+1})v_{i+1} - \hat{p}_{i+1}(v_{\leq i+1})} \\
&= \hat{p}_i(v_{\leq i}) + v_i x_i(v_{\leq i}) - \hat{p}_i(v_{\leq i}) \\
&= v_i x_i(v_{\leq i}),
\end{align*}
\noindent where the second inequality followed from the tightness of the utility bounds.
\end{proof}
Let us explain in more detail how a mechanism that charges the surplus on all days except for the last day may be a feasible solution to the original problem.  We first establish incentive compatibility.  Take a solution $(x,\hat{p})$ to the adjusted problem with binding utility bounds~\eqref{eq:adjustedIR} and its corresponding solution $(x,p)$ to the original problem.  At each day $i<k$, the mechanism $(x,p)$ gives each type the stage utility of zero since the type pays its surplus.  However, periodic incentive compatibility is satisfied by carefully setting the expected utility from future transactions.  In particular, for some $i$ and $v_{< i}$, consider the expected utility from future transactions that $v_i$ obtains by reporting $v'_i$,

\begin{align*}
U_i(v_{< i}, v'_i) &= \expect[v_{i+1}]{x_{i+1}(v_{< i},v'_i,v_{i+1})v_{i+1} - \hat{p}_{i+1}(v_{<i},v'_i,v_{i+1}) } \\
&= v'_ix_i(v_{\leq i}, v'_i) - \hat{p}_i(v_{\leq i}, v'_i),
\end{align*}
\noindent where the first equality is from~\eqref{eq:recursiveutility}, and the second equality followed from the tightness of the utility bounds. Since the agent obtains zero stage utility, the report will be chosen to maximize the utility from future transactions $U_i(v_{< i}, v'_i)$, and by~\eqref{eq:adjustedIC} the utility-maximizing report is a truthful report $v'_i = v_i$.  The mechanism is clearly ex-post IR since the stage utility is zero at each day $i<k$, and non-negative at the last day.  It is not immediately clear, however, that restricting to such mechanism is without loss of generality for the purpose of maximzing revenue, which is the reason we study the adjusted problem, concluding with \autoref{lem:martingalemechanism} which states that focusing on such constructions is without loss of generality.

Given the above analysis of the structure of the optimal solution, the rest of this section reduces the adjusted problem \eqref{eq:alternative formulation} into the \emph{surplus-utility-tradeoff} problem defined below.

\begin{definition}\label{def:revenue utility tradeoff general}
The \emph{surplus-utility-tradeoff} problem is parameterized by a single-dimensional distribution $f$, a utility bound $c \in \reals$, and a tradeoff function $g: \reals \rightarrow \reals$ and is defined as follows
\begin{align*}
\max_{x,p}  &\expect[v\sim f]{vx(v)+ g(vx(v)-p(v))}\\
\text{s.t., IC: } & vx(v) - p(v) \geq vx(v') - p(v') \\
& c \geq \expect[v \sim f]{vx(v)-p(v)}.
\end{align*}
A \emph{tight} surplus-utility-tradeoff problem is a surplus-utility-tradeoff problem in which the bound on expected utility must be tight, that is, $c = \expect[v \sim f]{vx(v)-p(v)}$.
\end{definition}

As an example, a special case of the above problem is when $c = +\infty$; and $g(u) = -u$ if $u \geq 0$, $g(u) = -\infty$ if $u<0$.  Note that the fact that $g(u) = -\infty$ for $u<0$ implies that the optimum solution must satisfy $vx(v)-p(v) \geq 0$ almost everywhere, and the constraint $+\infty  \geq \expect[v \sim f]{vx(v)-p(v)}$ is irrelevant.  By definition of function $g$, and subject to the constraint that $u(v) \geq 0$, the objevtice is to maximize $\expect[v\sim f]{vx(v)+ g(vx(v)-p(v))} = \expect[v \sim f]{p(v)}$.  As a result, this special case of the surplus-utility-tradeoff problem is equivalent to the standard monopoly pricing problem and the special case of our problem with $k=1$ (see \autoref{sec:myerson}).

\begin{example}[surplus-utility-tradeoff with Equal Revenue Distribution]\label{ex:equal revenue tradeoff}
Consider a surplus-utility-tradeoff problem where the distribution is the equal revenue distribution, and a tradeoff function $g$ that is $g(u) = -\infty$ for all $u<0$, and satisfies $g(u) = -u + h(u)$, where $0\leq h'(u) \leq 1$, $h(u) \geq 0$ over the range $u\in [0,\infty)$.  Suject to $u(v) \geq 0$, which is enforced by the assumption that $g(u) = -\infty$ for all $u<0$, the objective is to maximize $\expect[v\sim f]{vx(v) + g(u(v))} = \expect[v\sim f]{p(v) + h(u(v))}$. By Myerson's analysis (\autoref{ex:equal revenue}), the expected revenue of any mechanism satisfying the IC condition is $p(1)$.  As a result, the optimal mechanism solves
\begin{align}
\max_{x,p}  & \expect[v\sim f]{h(vx(v)-p(v))}+ p(1) \label{eq:rev utility tradeoff single}\\
\text{s.t., IC: } & vx(v) - p(v) \geq vx(v') - p(v')\nonumber \\
& c \geq \expect[v \sim f]{vx(v)-p(v)}.\nonumber
\end{align}
Consider $c$ large enough that the last constraint is irrelevant.  The structure of $g$ implies that the optimal solution must satisfy $u(v)\geq 0$ for all $v$, and in particular, $p(1) \leq x(1) \leq 1$.  Since $u'(v) \leq 1$ (\autoref{lem:myerson's lemma}), we must have $u(v) \leq v-1 + u(1)\leq v-p(1)$.  Thus, the solution to the above problem is at most $\expect[v\sim f]{h(v-p(1) )}+ p(1)$.  Since $h'(u)\leq 1$ for all $u\geq 0$, the maximum value of $\expect[v\sim f]{h(v-p(1) )}+ p(1)$ is achieved by setting $p(1)$ as large as possible, which is at most 1.  The solution to the problem is therefore at most $\expect[v\sim f]{h(v-1)}+ 1$, which is achieved by setting $x=1$ and $p=1$, that is, posting a price of $1$ for the item which is accepted by all types.  Notice that this analysis generalizes the analysis of \autoref{ex:equal revenue}.
\end{example}

We will next define a recursive family of functions, the \emph{cumulative tradeoff functions}, which will be later used to characterize optimal solutions to the adjusted problem in \autoref{lem:recursive}.

\begin{definition}[Cumulative Tradeoff Functions]\label{def:cum tradeoff}
Given $f_1,\ldots,f_k$, define the \emph{cumulative tradeoff functions} $\hat{g}_{k}(\cdot),\ldots,\hat{g}_{1}(\cdot)$ recursively as follows. For all $i$, $\hat{g}_{i}(c)$ is set to be the value of the solution to the \emph{tight} revenue utility tradeoff problem of Definition~\ref{def:revenue utility tradeoff general} for distribution $f_i$, utility bound $c$, and tradeoff function $g(u) = \hat{g}_{i+1}(u)$, when $i\geq 2$, and $g(u) = \hat{g}_2(u) - u$, when $i=1$. We also let $(X^c_i,P^c_i)$ be the corresponding optimal mechanism. For the above purposes, we take $\hat{g}_{\numitems+1}(u)=0$ if $u\geq 0$, and $\hat{g}_{\numitems+1}(u)=-\infty$ if $u<0$.
\end{definition}

The following lemma shows that the cumulative tradeoff functions $\hat{g}_i(\cdot)$ from Definition~\ref{def:cum tradeoff} can be used to capture the \emph{continuation value} of the adjusted dynamic program \eqref{eq:alternative formulation} for any choices made for the allocation and payment rules in a prefix of the periods.

\begin{lemma}\label{lem:continuation}
For any $j\leq k$ and $(y_1,\hat{q}_1), \ldots,(y_j,\hat{q}_j)$, the optimum value of \eqref{eq:alternative formulation} subject to the extra constraint that $(x_1,\hat{p}_1) = (y_1,\hat{q}_1), \ldots,(x_j,\hat{p}_j)=(y_j,\hat{q}_j)$, if the problem remains feasible, is equal to
\begin{align}
\expect[v_1,\ldots,v_j]{\hat{p}_1(v_1) + \left(\sum_{2\le i\leq j} x_i(v_{\leq i}) v_i\right) + \hat{g}_{j+1}(v_jx_j(v_{\leq j}) - \hat{p}_j(v_{\leq j}))} \label{eq:continuation}.
\end{align}
\end{lemma}
\begin{proof}
The proof is by induction, from $j=k$ to $j=1$.  The base of the induction trivially holds since when $j=k$, the above expression is equal to \eqref{eq:alternative formulation} if $v_kx_k(v_{\leq k}) - \hat{p}_k(v_{\leq k}) \geq 0$ for all $v_{\leq k}$ (recall that $\hat{g}_{k+1}(c)=0$ if $c\geq 0$ and $\hat{g}_{k+1}(c)=-\infty$ otherwise).  We show that the claim holds for $j-1$ assuming that it holds for $j$.  By induction hypothesis, the value of \eqref{eq:alternative formulation} subject to the extra constraint that $x_i = \hat{y}_i$ and $\hat{p}_i = \hat{q}_i$ for all $i\leq j$ is equal to \eqref{eq:continuation}. Now consider the value of \eqref{eq:alternative formulation} subject to the constraint that $x_i = \hat{y}_i$ and $\hat{p}_i = \hat{q}_i$ for all $i\leq j-1$.  Using the induction hypothesis, the problem is
\begin{align}
\max_{\hat{y}_j,\hat{q}_j}  &\expect{\hat{q}_1(v_1) + (\sum_{2\le i\leq j} \hat{y}_i(v_{\leq i}) v_i) + \hat{g}_{j+1}(v_j \hat{y}_j(v_{\leq j}) - \hat{q}_j(v_{\leq j}))} \label{eq:dayiproblem}\\
\text{s.t., } & v_j \hat{y}_j(v_{\leq j}) - \hat{q}_j(v_{\leq j})  \geq v_j \hat{y}_j(v_{< j}, v'_j) - \hat{q}_j(v_{<j},v'_j) & \forall v_{\leq j}, v'_j \nonumber \\
& v_{j-1} \hat{y}_{j-1}(v_{\leq j-1}) - \hat{q}_{j-1}(v_{\leq j-1}) \geq \expect[v_j]{\hat{y}_{j}(v_{\leq j})v_{j} - \hat{q}_{j}(v_{\leq j})}& \forall v_{\leq j-1}. \nonumber
\end{align}
By definition of the surplus-utility-tradeoff problem \eqref{eq:rev utility tradeoff single} and the cumulative tradeoff functions (\autoref{def:cum tradeoff}), the value of the above problem is 
\begin{align*}
\expect{\hat{q}_1(v_1) + (\sum_{2\le i< j} \hat{y}_i(v_{\leq i}) v_i) + \hat{g}_{j}(v_{j-1} \hat{y}_{j-1}(v_{\leq {j-1}}) - \hat{q}_{j-1}(v_{\leq {j-1}}))}.
\end{align*} 
\end{proof}

The following lemma uses the characterization of continuation value of the dynamic program provided in \autoref{lem:continuation} to state a structural property of the solutions $(X^c_i,P^c_i)$ to \eqref{eq:alternative formulation} in terms of the cumulative tradeoff functions.  Our characterization of optimal solutions is based on the following lemma.

\begin{lemma}\label{lem:recursive}
Consider a mechanism $(x,\hat{p})$ defined recursively from $i = 1$ to $k$, given solutions $\hat{g}_i(c)$ and $(X^c_i,P^c_i)$ to the cumulative tradeoff problems of \autoref{def:cum tradeoff},  as follows.   The mechanism at day 1, $(x_1,\hat{p}_1)$ is equal to $(X^{c_0}_1,P^{c_0}_1)$, where $c_0$ maximizes $\hat{g}_1(c)$.  The mechanism at day $i \geq 2$ is $x_i(v_{\leq i}) = X^{c_{i-1}}_i(v_i)$, and $\hat{p}_i(v_{\leq i}) = P^{c_{i-1}}_i(v_i)$ for $c_{i-1} = v_{i-1}x_{i-1}(v_{<i}) - \hat{p}_{i-1}(v_{<i})$.  The mechanism $(x,\hat{p})$ is an optimal solution to the adjusted problem.
\end{lemma}
\begin{proof}
The lemma follows from \autoref{lem:continuation} and expression \eqref{eq:continuation}  as follows.  Consider an optimal solution $(x^*_i,\hat{p}^*_i)$  to  problem \eqref{eq:alternative formulation}.  Fixing $\hat{y}_1=x^*_1,\hat{q}_{1}=\hat{p}^*_{1}$ to $\hat{y}_{i-1}=x^*_{i-1},\hat{q}_{i-1}=\hat{p}^*_{i-1}$, $(x^*_i,\hat{p}^*_i)$ must be the solution to \eqref{eq:dayiproblem}.  That is, $(x^*_i,\hat{p}^*_i)(v_{\leq i}) = (x_{v_{<i}},p_{v_{<i}})(v_i)$, where $(x_{v_{<i}},p_{v_{<i}})$ is the solution to the revenue utility tradeoff problem for distribution $f_i$, utility bound $v_{i-1}x^*_{i-1}(v_{<i}) - \hat{p}^*_{i-1}(v_{<i})$, and the tradeoff function $g(x) = \hat{g}_{i+1}(x)+x$ for $i\geq 2$, and $g(x) = \hat{g}_{2}(x)$ for $i=1$.
\end{proof}

The above lemma suggests a procedure to characterize the solution $(x^*,\hat{p}^*)$ to the adjusted problem: recursively (from $k$ to $1$) solve for all cumulative tradeoff functions $\hat{g}_i(c)$ and mechanisms $(X^c_i,P^c_i)$ as per Definition~\ref{def:cum tradeoff}; recursively (from $1$ to $k$), define $(x^*_i,\hat{p}^*_i)$ given $(x^*_{i-1},\hat{p}^*_{i-1})$ as specified by Lemma~\ref{lem:recursive}.  Finally, 
use \autoref{lem:changeofvariables} to convert the adjusted mechanism $(x^*_i,\hat{p}^*_i)$ to an optimal mechanism $(x^*,p^*)$ for the original problem.   Since $(x^*_i,\hat{p}^*_i)$ satisfies all the utility bounds with equality, the payment at each day $i < k$ must be equal to $v_i x^*_i(v_{\leq i})$.

\begin{theorem}\label{thm:characterization}
An optimal mechanism for the original problem is characterized as follows: 
\begin{enumerate}
\item[0.] (Pre-processing) Recursively (from $\numitems$ to $1$) define the cumulative tradeoffs $\hat{g}_i(c)$ and mechanisms $(X^c_i,P^c_i)$ for all $i$ and $c$ as solutions to the surplus utility tradeoff problem (\autoref{def:cum tradeoff}).  Set $c_0$ equal to the maximizer of $\hat{g}_1(c)$.
\item[1.] At each day $i \ge 1$, if the buyer reports $v_i$, he is allocated with probability $X^{c_{i-1}}_i(v_i)$, pays $v_iX^{c_{i-1}}_i(v_i)$ if $i<k$ or $P^{c_{i-1}}_i(v_i)$ if $i=k$. We also set $c_i = v_i X^{c_{i-1}}_i(v_i) - P^{c_{i-1}}_i(v_i)$.
\end{enumerate}
\end{theorem}

Let us now discuss the computational implications of the above characterization.  The optimum mechanism can be calculated exactly using the above transformation and recursion. However, the recursive computation requires solving and listing the values of function $\hat{g}_i$ over a continuous domain (of all positive utility bounds $c$) in order to solve for $\hat{g}_{i-1}$. The benefit of the formulation is that the information that is passed from round $i$ to round $i-1$ is only a scalar function as opposed to a multi-variate function resulting from the trivial formulation. Moreover, notice that the functions $\hat{g}_i$ arising in Definition~\ref{def:cum tradeoff} are concave and that, in order to solve the revenue-utility tradeoff problem defining $\hat{g}_i$, only oracle access to the function $\hat{g}_{i+1}$ is required. So finding each $\hat{g}_i(c)$ value given oracle access to function $\hat{g}_{i+1}$ amounts to a convex program. Indeed, we can exploit this observation to obtain a Fully Polynomial Time Approximation Scheme (FPTAS) in the case where the support of the type distributions $f_1,\ldots,f_k$ is discrete. See Appendix~\ref{app:FPTAS}.

%

\begin{example}[The Equal Revenue Distribution at Stage 1]\label{ex:equal revenue day 1}
Consider a 2-stage problem where the first distribution is an equal revenue distribution.  By \autoref{lem:recursive}, the mechanism at day $1$ optimizes $\expect[v_1]{\hat{p}_1(v_1) + \hat{g}_2(u_1(v_1))}$.  By \autoref{ex:equal revenue tradeoff}, the solution at day 1 satisfies $(x^*_1,\hat{p}^*_1) = (1,p(1))$. This analysis suggests that the following mechanism is optimal.
\begin{enumerate}
\item Allocate the item at stage 1; charge $v_1$.
\item At stage 2, the allocation and payment $(x^*_2(v_1,v_2), p^*_2(v_1,v_2))$ are the solutions to the the problem of maximizing the expected surplus for distribution $f_2$ subject to tight utility bound $c=v_1-p(1)$.   In \autoref{sec:utility constrained surplus}, we show that the solution to this problem is to randomize over at most two posted prices that give the agent expected utility equal to $v_1 - p(1)$.
\end{enumerate} 
\end{example}

\subsection{A Simple 2-approximation}\label{subsec:2approx}
We now describe a simple 2-approximately optimal single-bidder dynamic mechanism.  At the core of the analysis is using the adjusted formulation \eqref{eq:alternative formulation} to identify an upper bound on revenue.

\begin{lemma}\label{lem:separaterevenue}
Consider any feasible solution $(x,\hat{p})$ to the adjusted problem.  The objective value of the solution is at most
\begin{align*}
\Big[\sum_{i \geq 1} \expect{\max(\phi_i(v_i),0)}\Big]  + \Big[\sum_{i \geq 2} \expect{v_ix_i(v_{\leq i}) - \hat{p}_i(v_{\leq i})} + \sum_{i \geq 1} \expect{\hat{p}_i(v_{<i},0)} \Big]
\end{align*}
\end{lemma}
\begin{proof}
Consider the objective value of a feasible solution $(x,\hat{p})$,
\begin{align}
\expect{\hat{p}_1(v_1) + \sum_{i=2}^k x_i(v_{\leq i}) v_i} &= \expect{\hat{p}_1(v_1) + \sum_{i=2}^k x_i(v_{\leq i}) v_i - \hat{p}_i(v_{\leq i}) + \hat{p}_i(v_{\leq i})} \nonumber \\
&= \sum_{i \geq 1} \expect{\hat{p}_i(v_i)} + \sum_{i=2}^k \expect{x_i(v_{\leq i}) v_i - \hat{p}_i(v_{\leq i})} \nonumber
\intertext{Using the characterization of incentive compatibility, Equation~\eqref{eq:virtual value}, for each $i \geq 1$ and $v_{<i}$,}
\expect[v_i]{\hat{p}_i(v_{\leq i})} = \expect[v_i]{x_i(v_{\leq i})\phi(v_i)} + \hat{p}_i(v_{< i},0) &\leq \expect[v_i]{\max(\phi_i(v_i),0)} + \hat{p}_i(v_{< i},0).\label{eq:payment bound}
\intertext{Note that establishing \eqref{eq:payment bound} proves the lemma, since the objective value can be upper bounded as follows}
\sum_{i \geq 1} \expect{\hat{p}_i(v_i)} + \sum_{i=2}^k \expect{x_i(v_{\leq i}) v_i - \hat{p}_i(v_{\leq i})} &\leq \sum_{i \geq 1} \expect{\max(\phi_i(v_i),0)} + \sum_{i \geq 1}  \expect{\hat{p}_i(v_{< i},0)} \\ \nonumber& + \sum_{i=2}^k \expect{x_i(v_{\leq i}) v_i - \hat{p}_i(v_{\leq i})}, \nonumber
\end{align}
\noindent as claimed.
\end{proof}

We will next design two mechanisms, achieving respectively revenue that dominates the maximum value that each of the two terms in the statement of \autoref{lem:separaterevenue} can take. The more involved part of the analysis studies the maximum value that the second term can take, that is  maximizing
\begin{align}
\sum_{i \geq 2} \expect{v_ix_i(v_{\leq i}) - \hat{p}_i(v_{\leq i})} + \sum_{i \geq 1} \expect{\hat{p}_i(v_{<i},0)},\label{eq:separated objective}
\end{align}
\noindent subject to the feasibility conditions \eqref{eq:adjustedIC} and utility bounds \eqref{eq:adjustedIR} of the adjusted problem {forced to be tight for $i<k$}.  Similar to \autoref{subsec:char}, we first show that the optimal solution can be characterized recursively using tradeoff functions that capture the continuation value of the problem.  We will later show that the optimizers of the recursive problem have a simple form.

\begin{definition}[Cumulative Tradeoff Functions for objective \eqref{eq:separated objective}]\label{def:separated tradeoffs}
Given $f_1,\ldots,f_k$, define the tradeoff functions $\hat{h}_{k}(\cdot),\ldots,\hat{h}_{1}(\cdot)$ recursively as follows. Define $\hat{h}_{k+1}(c) = 0$ for all $c \ge 0$, and $\hat{h}_{k+1}(c) = -\infty$ otherwise. Recursively, for all $i$, $\hat{h}_i(c)$ is the optimal solution  to the problem of finding mechanism $(x_i,p_i)$ to maximize
\begin{align*}
&\expect[v_i \sim f_i]{\hat{h}_{i+1}(v_ix_i(v_i) - p_i(v_i))} + c + p_i(0),
\intertext{subject to incentive compatibility of the mechanism, and a tight bound on the expected utility of the mechanism as follows:}
c &= \expect[v_i \sim f_i]{v_ix_i(v_i) - p_i(v_i)}.
\end{align*}
\noindent Let $(Y^c_i,Q^c_i)$ be the mechanism that achieves the optimum value.
\end{definition}

The following lemma shows that the tradeoff functions $\hat{h}$ defined in \autoref{def:separated tradeoffs} capture the continuation value of the dynamic program~\eqref{eq:separated objective}, and the mechanisms $(Y^c_i,Q^c_i)$ characterize the optimal solution to the dynamic program.  The proof is similar to \autoref{lem:recursive} and is omitted.

\begin{lemma} \label{lem:characterization of dummy solution}
Consider a mechanism $(x,\hat{p})$ defined recursively from $i = 1$ to $k$, given solutions $\hat{h}_i(c)$ and $(Y^c_i,Q^c_i)$ to the cumulative tradeoff problems of~\autoref{def:separated tradeoffs},  as follows.   The mechanism at day 1, $(x_1,\hat{p}_1)$ is equal to $(Y^{c_0}_1,Q^{c_0}_1)$, where $c_0$ maximizes $\hat{h}_1(c)-c$.  The mechanism at day $i \geq 2$ is $x_i(v_{\leq i}) = Y^{c_{i-1}}_i(v_i)$, and $\hat{p}_i(v_{\leq i}) = Q^{c_{i-1}}_i(v_i)$ for $c_{i-1} = v_{i-1}x_{i-1}(v_{<i}) - \hat{p}_{i-1}(v_{<i})$.  Then the mechanism $(x,\hat{p})$ is an optimal solution to the problem of maximizing \eqref{eq:separated objective} subject to \eqref{eq:adjustedIC} and \eqref{eq:adjustedIR}.
\end{lemma}

We will next show in \autoref{lem:simple form} that the functions $(Y^c_i,Q^c_i)$, which characterize the optimal solution for objective \eqref{eq:separated objective} as per Lemma~\ref{lem:characterization of dummy solution}, have simple forms.  We first establish two technical properties in \autoref{lem:aux1} and \autoref{lem:single crossing}, which we use to prove \autoref{lem:simple form}.  

The first technical lemma shows that functions $\hat{h}$ are concave.
\begin{lemma}\label{lem:aux1}
For each $i$ and $c$, the function $\hat{h}_i(c)$ defined in \autoref{def:separated tradeoffs} is concave.
\end{lemma}
\begin{proof}
We prove the claim inductively. The function $\hat{h}_{k+1}$ is trivially concave. Consider $c_1$ and $c_2$, and their corresponding optimal mechanisms $(Y^{c_1}_i,Q^{c_1}_i)$ and $(Y^{c_2}_i,Q^{c_2}_i)$.  Note that the average of these two mechanisms $(Y,Q)$ is incentive compatible, and has expected utility equal to $(c_1+c_2)/2$.  Thus, $(Y,Q)$ is a feasible solution to the problem whose optimum is $\hat{h}_i((c_1+c_2)/2)$, as per Definition~\ref{def:separated tradeoffs}.  In addition,
\begin{align*}
&c+Q(0) = \frac{1}{2} (c_1 + Q^{c_1}_i(0)) + \frac{1}{2}(c_2 + Q^{c_2}_i(0)),
\intertext{and,}
&\expect{\hat{h}_{i+1}(v_iY(v_i) - Q(v_i))} \geq \frac{1}{2}\expect{\hat{h}_{i+1}(v_iY^{c_1}_i(v_i) - Q^{c_1}_i(v_i))} + \frac{1}{2}\expect{\hat{h}_{i+1}(v_iY^{c_2}_i(v_i) - Q^{c_2}_i(v_i))},
\intertext{by concavity of $h_{i+1}$.  As a result,}
&\hat{h}_i((c_1+c_2)/2) \geq (c_1+c_2)/2 + Q(0) + \expect{\hat{h}_{i+1}(v_iY(v_i) - Q(v_i))} \geq \frac{1}{2}\hat{h}_i(c_1) + \frac{1}{2}\hat{h}_i(c_2).
\end{align*}
\end{proof}

The second technical lemma provides a set of conditions that allows us to compare the value of $\expect[v_i \sim f_i]{\hat{h}_{i+1}(v_ix_i(v_i) - p_i(v_i))}$ achieved by two different solutions.

\begin{lemma}\label{lem:single crossing}
Consider a concave function $h$ and a two monotone non-decreasing functions $u_1, u_2: \reals \rightarrow \reals$ satisfying $\expect{u_1(z)}=\expect{u_2(z)}$ and with a threshold $z^0$ such that $u_1(z) \geq u_2(z)$ for all $z \leq z^0$, and $u_1(z) \leq u_2(z)$ otherwise.  Then, $\expect{h(u_1(z))} \geq \expect{h(u_2(z))}$.
\end{lemma}
\begin{proof}
	For any $z$, define $\delta(z) = u_1(z) - u_2(z)$, and note for future reference that $\expect{\delta(z)} = 0$ by lemma's assumption.  We will first argue that by concavity of $h$, 
	\begin{align}
	h(u_1(z))-h(u_2(z)) \geq h(u_1(z^0)) - h(u_1(z^0)-\delta(z)).\label{eq:aux5}
	\end{align}
	To prove the above inequality, consider two cases.  If $z \leq z^0$, then $u_2(z)\leq u_1(z)\leq u_1(z^0)$ by lemma's assumptions.  Concavity of $h$ implies that $h(u_1(z)) - h(u_2(z)) \geq h(u_1(z^0)) - h(u_1(z^0)-\delta(z))$.  Similarly, if $z \geq z^0$, then $u_1(z^0) \leq u_1(z)\leq u_2(z)$ by lemma's assumptions.  In this case, by concavity of $h$ we have $h(u_2(z))-h(u_1(z)) \leq h(u_1(z^0) - \delta(z)) - h(u_1(z^0))$, which implies~\eqref{eq:aux5}.  We now have
	
	\begin{align*}
	\expect{h(u_1(z))} - \expect{h(u_2(z))} \geq h(u_1(z^0)) - \expect{h(u_1(z^0)-\delta(z))} \geq 0,
	\end{align*}
	where the first inequality followed from \eqref{eq:aux5}, and the second inequality followed from Jensen's inequality and concavity of $h$: since $\expect{\delta(z)} = 0$, then $\expect{h(u_1(z^0)-\delta(z))} \leq h(\expect{u_1(z^0)-\delta(z)}) = h(u_1(z^0))$.

\end{proof}

Now we turn to our main lemma that states that the solutions  $(Y^c_i,Q^c_i)$ to the tradeoff problems of \autoref{def:separated tradeoffs} have simple forms.  In particular, the allocation and payment functions are constants that do not depend on $v_i$, but instead depend only on $c$ and $i$.

\begin{lemma}\label{lem:simple form}
For each $i$ and $c$, $Y_i^c(v_i)$ and $Q_i^c(v_i)$ are constant functions of $v_i$.
\end{lemma}
\begin{proof}
Consider the optimal mechanisms $(Y_i^c, Q^c_i)_{i,c}$ of Definition~\ref{def:separated tradeoffs}. For some fixed $i$ and $c$, consider the mechanism $(\tilde{Y}^c_i,\tilde{Q}^c_i)$ obtained from $(Y^c_i,Q^c_i)$ as follows: define $\tilde{Y}^c_i(v_i) = (c+Q^c_i(0))/\expect{v_i}$ and $\tilde{Q}^c_i(v_i) = Q^c_i(0)$.  That is, the mechanism offers all types a constant probability of allocation $(c+Q^c_i(0))/\expect{v_i}$ for a constant payment $Q^c_i(0)$.  The constants are adjusted to maintain the tightness of the utility bound, since the expected utility of the mechanism $(\tilde{Y}^c_i,\tilde{Q}^c_i)$ is
\begin{align*}
\expect{v_i \tilde{Y}^c_i(v_i) - \tilde{Q}^c_i(v_i)}  = \expect{v_i}\frac{c+Q^c_i(0)}{\expect{v_i}} - Q^c_i(0) = c.
\end{align*}
Notice that the aforedescribed mechanism is well-defined as $(c+Q^c_i(0))/\expect{v_i} \in [0,1]$. To see this, note that by characterization of incentive compatibility, the utility of a type $v_i$ in mechanism $(Y_i^c, Q^c_i)$ is at most $v_i - Q^c_i(0)$.  Since the expected utility of the mechanism is $c$, we have $c \leq \expect{v_i - Q^c_i(0)} = \expect{v_i} - Q^c_i(0)$.  As a result, $(c+Q^c_i(0))/\expect{v_i} \leq 1$ as claimed.  Similarly, since the utility of type $v_i$ is at least $-Q^c_i(0)$, we must have $c \geq -Q^c_i(0)$, and thus $(c+Q^c_i(0))/\expect{v_i} \geq 0$.

Next we show that the mechanism $(\tilde{Y}^c_i,\tilde{Q}^c_i)$ is at least as good as $(Y^c_i,Q^c_i)$ in terms of objective value.  By applying this argument to all $i$ and $c$, the lemma follows.  Let us compare the values obtained by $(\tilde{Y}^c_i,\tilde{Q}^c_i)$,
\begin{align*}
&\expect{\hat{h}_{i{+1}}(v_i\tilde{Y}^c_i(v_i) - \tilde{Q}^c_i(v_i))} + c + \tilde{Q}^c_i(0),
\intertext{and $(Y^c_i,Q^c_i)$,}
 &\expect{\hat{h}_{i{+1}}(v_i{Y}^c_i(v_i) - {Q}^c_i(v_i))} + c + {Q}^c_i(0).
 \intertext{Since $\tilde{Q}^c_i(0) = {Q}^c_i(0)$, we only need to prove that}
& \expect{\hat{h}_{i{+1}}(v_i\tilde{Y}^c_i(v_i) - \tilde{Q}^c_i(v_i))} \geq \expect{\hat{h}_{i{+1}}(v_i{Y}^c_i(v_i) - {Q}^c_i(v_i))}.
 \end{align*} 
\noindent Define $\tilde{u}^c_i$ to be the utility function of mechanism $(\tilde{Y}^c_i,\tilde{Q}^c_i)$, that is, $\tilde{u}^c_i(v_i) =v_i\tilde{Y}^c_i(v_i) - \tilde{Q}^c_i(v_i)$, and $u^c_i$ to be the utility function of mechanism $(Y^c_i,Q^c_i)$, that is, ${u}^c_i(v_i) =v_i{Y}^c_i(v_i) - {Q}^c_i(v_i)$.   We will argure that $\tilde{u}^c_i$ and $u^c_i$ satisfy the conditions of \autoref{lem:single crossing} (by setting $u_1 = \tilde{u}^c_i$ and $u_2 = {u}^c_i$), and thus $\expect{\hat{h}_{i{+1}}(\tilde{u}^c_i(v_i))} \geq  \expect{\hat{h}_{i{+1}}(u^c_i(v_i))}$, {since $\hat{h}_{i{+1}}$ is concave by Lemma~\ref{lem:aux1}}.  We have already argued that $\expect{\tilde{u}^c_i(v_i)} = \expect{u^c_i(v_i)} = c$.  We will next argue that there exists a $v^0_i$, such that $\tilde{u}^c_i(v_i) \geq u^c_i(v_i)$ when $v_i \leq v^0_i$, and $\tilde{u}^c_i(v_i) \leq u^c_i(v_i)$ otherwise.   Since $\tilde{u}^c_i$ and $u^c_i$ have the same expectation and $\tilde{u}^c_i(0) = u^c_i(0) = -{Q}^c_i(0)$, they must cross at least once at a point $v^0_i > 0$ (otherwise, one function is pointwise higher than the other, contradicting the fact that they have the same expectation).   Convexity of the utility function $u^c_i$ implies that it must be below a line that connects $u^c_i(0)$ to $u^c_i(v^0_i)$.  This line is in fact $\tilde{u}^c_i$.  That is, $u^c_i$ crosses $\tilde{u}^c_i$ from below at $v^0_i$, and since it is convex, it must stay above $\tilde{u}^c_i$ for all $v_i \geq v^0_i$. 
\end{proof}

By \autoref{lem:simple form},  $(Y^c_i,Q^c_i)$ are simply constants that do not depend on $v_i$.  Note that selecting $Y^c_i$ will uniquely identify $Q^c_i$ through the expected utility constraint, that is $Q^c_i = Y^c_i \expect{v_i} - c$.   We therefore redefine $\hat{h}_i(c)$ simply as the optimal solution to
\begin{align}
\max_{0 \leq Y \leq 1} \expect{\hat{h}_{i+1}(Y(v_i - \expect{v_i}) + c)} + Y\expect{v_i},\label{eq:simply problem}
\end{align}
\noindent and let $Y^c_i$ simply refer to the optimizer of the above problem.

\begin{theorem}\label{thm:2 approx characterization}
Running each of the following two mechanisms with probability a half gives a 2-approximation to the optimal revenue: 
\begin{enumerate}
\item[0.] (Pre-processing) Define the tradeoffs $\hat{h}_{k}(c),\ldots,\hat{h}_{1}(c)$ and allocation probabilities $Y^c_i$ for all $c$ and for all $i$ recursively as follows. Define $\hat{h}_{k+1}(c) = -\infty$, for all $c<0$, and $\hat{h}_{k+1}(c) = 0$, for all $c \ge 0$. Recursively for all $i$, and for all $c$, $\hat{h}_i(c)$ is the optimal solution to \eqref{eq:simply problem} and $Y^c_i$ is the optimizer.  Set $c_0$ equal to the maximizer of $\hat{h}_1(c)-c$.
\item[1.] (Mechanism 1) At each day $i$, ignore the history and offer the optimal monopoly price for item $i$.
\item[2.] (Mechanism 2) At each day $i$, the buyer reports $v_i$, is allocated with probability $Y^{c_{i-1}}_i$, pays $v_iY^{c_{i-1}}_i$ if $i<k$ or $Y^{c_{i-1}}_i \expect{v_i} - c_{i-1}$ if $i=k$.  Set $c_i = Y^{c_{i-1}}_i(v_i -\expect{v_i}) + c_{i-1}$.

\end{enumerate}
\end{theorem}
\begin{proof}
Consider the upper bound provided in \autoref{lem:separaterevenue} on optimal revenue.  The above two mechanisms bound the first and second terms in \autoref{lem:separaterevenue}, respectively.  Mechanism 1 obtains revenue equal to 
\begin{align*}
\sum_{i \geq 1} \expect{\max(\phi_i(v_i),0)}.
\end{align*}
\noindent By \autoref{def:separated tradeoffs} and \autoref{lem:simple form}, Mechanism 2 achieves a revenue that bounds the second term of the upper bound in \autoref{lem:separaterevenue}. {Note that, in describing Mechanism 2, we have translated from a solution to the adjusted problem to a solution to the original formulation.}
\end{proof}

\subsection{The Multi-agent Problem}
This section extends the analysis of \autoref{subsec:char} to designing sequential auctions with multiple agents.  The key step is defining the right change of variables such that the problem mirrors the single-agent problem in \autoref{subsec:char}.  Once the right formulation is identified, the analysis extends to multiple agents straightforwardly.

For agents $1$ to $m$, assume that the value of each agent $\agent$ on each day $i=1,\ldots,k$, $v_i^\agent$, is drawn independently of other values and other days from a distribution with density $f^\agent_i$.  For each $i$, let $v_i = (v_i^1,\ldots,v_i^m)$ be the vector of values of all agents at day $i$, and $v_{\leq i} = (v_1,\ldots,v_i)$ be the complete history of all values up to day $i$. A multi-agent mechanism specifies the allocation probability $x^\agent_i(v_{\leq i})$ and payment $p^\agent_i(v_{\leq i})$ of each agent $\agent$ at each step $i$ based on the current and history of reports of all agents.   The goal is to maximize
\begin{align}
\expect{\sum_i \sum_\agent p^\agent_i(v_{\leq i})},\label{eq:multi revenue}
\end{align}
subject to appropriately defined  incentive compatibility and ex-post individual rationality conditions.  In particular, the periodic incentive compatibility condition requires that on each day $i$ and for every agent $\agent$, given any history of values $v_{<i}$ of all agents and value $v^\agent_i$ of agent $\agent$ on the current day, agent $\agent$ maximizes her expected utility, with respect to today's value of other agents and future values of all agents including $\agent$, by truthfully reporting its type, i.e.

\begin{align}
&\expect[v_i^{-\agent},v_{i+1},\ldots,v_k]{\sum_{j\geq i} v_j^\agent x^\agent_j(v_{\leq j}) - p^\agent_j(v_{\leq j})} \\ &\geq \expect[v_i^{-\agent},v_{i+1},\ldots,v_k]{\sum_{j\geq i} v_j^\agent x^\agent_j(v_{\leq j},\hat{v}^\agent_i) - p^\agent_j(v_{\leq j},\hat{v}^\agent_i)}; \forall v_{<i}, v^\agent_i, \hat{v}^\agent_i,
\end{align}
where we sloppily use $(v_{\leq j},\hat{v}^\agent_i)$ to denote $v_{\leq j}$ with the value of agent $\agent$ on day $i$ replaced by $\hat{v}^\agent_i$.
\noindent  Ex-post individual rationality  requires that each agent's utility is non-negative at each stage and for any history,
\begin{align}
&v^\agent_i x^\agent_i(v_{\leq i}) - p^\agent_i(v_{\leq i}) \geq 0; \forall {v}_{\leq i}\label{eq:multi IR}
\end{align}

\noindent We next argue that we can replace the above condition with
\begin{align}
&\expect[v_i^{-\agent}]{v^\agent_i x^\agent_i(v_{\leq i}) - p^\agent_i(v_{\leq i})} \geq 0; \forall {v}_{<i}, v_i^\agent.\label{eq:multi IR2}
\end{align}
\noindent Note that the above condition is obtained by taking the expectation of~\eqref{eq:multi IR} over $v_i^{-\agent}$ and therefore is implied by it.  Conversely, if a mechanism satisfies the above condition, then by charging the agent $\expect[v_i^{-\agent}]{p^\agent_i(v_{\leq i})}/\expect[v_i^{-\agent}]{x^\agent_i(v_{\leq i})}$ whenever he is allocated, condition~\eqref{eq:multi IR} will be satisfied.

Finally, feasibility of the mechanism requires that at each stage only one item is allocated
\begin{align}
&\sum_\agent x^\agent_i(v_{\leq i}) \leq 1; \forall i,v_{\leq i}.
\end{align}
\noindent  Similar to \autoref{sec:dynamic mechanism design and optimal revenue-utility tradeoff}, we formulate the above problem in terms of an \emph{adjusted} payment function.  In particular, let $U_i^\agent(v_{\leq i})$ be the expected utility of the agents from all days after $i$, that is,
\begin{align*}
U_i^\agent(v_{\leq i}) = \expect[v_{i+1},\ldots,v_k]{\sum_{j> i} v_j^\agent x^\agent_j(v_{\leq j}) - p^\agent_j(v_{\leq j})}.
\end{align*}

Define the \emph{adjusted payment} $\hat{p}^\agent_i$ as a function of $v_{<i}$ and $v_i^\agent$ as follows
\begin{align}
\hat{p}^\agent_i(v_{<i},v_i^\agent) &:= \expect[v_i^{-\agent}]{p^\agent_i(v_{\leq i}) - U^\agent_i(v_{\leq i})}.\label{eq:adjusted p multi agent}
\end{align}
An analysis identical to that of \autoref{sec:dynamic mechanism design and optimal revenue-utility tradeoff} shows that the problem can be rewritten using the adjusted payment function in the following form.  We will refer to the following problem as the adjusted problem.\footnote{For the purposes of our formulation we take, for notational convenience, $x^\kappa_{k+1}$ and $\hat{p}^\kappa_{k+1}$ to be the zero functions, for all $\kappa$.}

\begin{align}
\max_{x,\hat{p}}  &\expect{\sum_\agent (\hat{p}^\agent_1(v^\agent_1) + \sum_{i \geq 2} v^\agent_ix^\agent_i(v_{\leq i}))} \label{eq:adjusted multi agent}\\
\text{s.t., $\forall i, \kappa,$ IC$^\agent_i$: }&\forall v_{< i},v^\agent_i,\hat{v}^\agent_i:   \expect[v^{-\agent}_i]{v^\agent_i x^\agent_i(v_{\leq i}) - \hat{p}^\agent_i(v_{< i},v_i^\agent)}  \geq \expect[v^{-\agent}_i]{v^\agent_i x^\agent_i(v_{\leq i}, \hat{v}^\agent_i) - \hat{p}^\agent_i(v_{< i},\hat{v}^\agent_i)} \label{eq:IC multi agent} \\
& \forall i,\kappa, {v_{< i}},v^\agent_i:  \expect[v_i^{-\agent}]{v^\agent_i x^\agent_i(v_{\leq i}) - \hat{p}^\agent_i(v_{< i},v_i^\agent)} \geq \expect[v_i^{-\agent},v_{i+1}]{v^\agent_{i+1}x^\agent_{i+1}(v_{\leq i+1}) - \hat{p}^\agent_{i+1}({v_{\leq i}},v_{i+1}^\agent)}\label{eq:utility bound multi agent}\\
&\forall i,v_{\leq i}: \sum_\agent x^\agent_i(v_{\leq i}) \leq 1.\nonumber
\end{align}

\begin{lemma}\label{lem:changeofvariables multi}
	A mechanism $(x,p)$ is a feasible solution to the original problem if and only if the mechanism $(x,\hat{p})$ defined via Equation~\eqref{eq:adjusted p multi agent}, $\hat{p}^\agent_i(v_{<i},v_i^\agent) := \expect[v_i^{-\agent}]{p^\agent_i(v_{\leq i}) - U^\agent_i(v_{\leq i})}$ is a feasible solution to the adjusted problem.   The revenue of $(x,p)$ in the original formulation is equal to the objective value of $(x,\hat{p})$ in the adjusted formulation~\eqref{eq:alternative formulation}.  In particular $(x,p)$ is an optimal solution to the original problem if and only if $(x,\hat{p})$ is an optimal solution to the adjusted problem.
\end{lemma}

A similar argument to that of \autoref{sec:dynamic mechanism design and optimal revenue-utility tradeoff} shows that without loss of generality, the optimal solution to the problem~\eqref{eq:adjusted multi agent} satisfies all the utility bounds~\eqref{eq:utility bound multi agent} with equality.  As a result, in the rest of this section we require that such inequalities are tight. We next make the following definition, which is analogous to \autoref{def:revenue utility tradeoff general}.

\begin{definition}\label{def:revenue utility tradeoff general multi}
	The \emph{multi-agent surplus-utility-tradeoff} problem is parameterized by single-dimensional distributions $f^1,\ldots,f^m$, utility bounds $(c^1,\ldots,c^m) \in \reals^m$, and a tradeoff function $g: \reals^m \rightarrow \reals$. The goal is to design functions $x^\agent(v)$ and $p^\agent(v)$, inducing $u^\agent(v^\agent):= \expect[v^{-\agent}]{v^\agent x^\agent(v)-p^\agent(v)}$, for each agent $\agent$ as follows
	\begin{align}
	&\max_{x,p} \expect[v]{\sum_\agent (x^\agent(v) v^\agent) + g\left(u^1(v^1), \ldots, u^m(v^m)\right)} \label{eq:multi agent tradeoff} \\
	\text{s.t.}, &\forall \kappa, v^\agent,\hat{v}^\agent: \expect[v^{-\agent}]{v^\agent x^\agent(v)-p^\agent(v)}\geq \expect[v^{-\agent}]{v^\agent x^\agent (v^{-\agent},\hat{v}^\agent)-p^\agent(v^{-\agent},\hat{v}^\agent)} \nonumber \\
	& \forall \agent: \expect[v^\agent]{u^\agent(v^\agent)} \leq c^\agent \nonumber\\
	&\forall v: \sum_{\agent} x^\agent(v) \leq 1.\nonumber
	\end{align}
	A \emph{tight} surplus-utility-tradeoff problem is a surplus-utility-tradeoff problem in which the bound on expected utility must be tight, that is, $\expect[v^\agent]{u^\agent(v^\agent)} = c^\agent$ for all $\agent$.
\end{definition}


Note that given oracle access to a concave function $g$, the above problem is convex. We can use techniques from \cite{AFHHM12} and \cite{CDW13b} to reformulate and efficiently solve the problem in terms of interim functions $x^\agent(v^\agent) := \expect[v^{-\agent}]{x^\agent(v)}$ and $p^\agent(v^\agent) := \expect[v^{-\agent}]{p^\agent(v)}$, and then map the solution back to an ex-post description of the problem.  

The following lemma defines cumulative tradeoff functions which will later be used to characterize the value of the dynamic problem, and parallels \autoref{def:cum tradeoff}.

\begin{definition}\label{def:cum tradeoff multi}
	Given distributions $f^\agent_i$ for all days $i$ and agents $\agent$, define the \emph{cumulative tradeoff functions} $\hat{g}_{k}(\cdot),\ldots,\hat{g}_{1}(\cdot)$ recursively as follows. For all $i$ and profile $\vec{c}=(c^1,\ldots,c^m)$, $\hat{g}_{i}(\vec{c})$ is set to be the value of the solution to the \emph{tight} revenue utility tradeoff problem of Definition~\ref{def:revenue utility tradeoff general multi} for distributions $f^1_i,\ldots,f^m_i$, utility bounds $\vec{c}$, and tradeoff function $g(\vec{u}) = \hat{g}_{i+1}(\vec{u})$, when $i\geq 2$, and $g(\vec{u}) = \hat{g}_2(\vec{u}) - \sum_{\agent}u^\agent$, when $i=1$. We also let $(X_i,P_i)({\vec{c}})$ be the corresponding optimal mechanism, and $U_i({\vec{c}})$ be its interim utility function, that is $U^\agent_i({\vec{c}},v_i^\agent) := \expect[v_i^{-\agent}]{v_i^\agent X^\agent_i({\vec{c}},v_i) - P^\agent_i({\vec{c}},v_i)}$. For the above purposes, we take $\hat{g}_{\numitems+1}(u)=0$ if $u\geq 0$, and $\hat{g}_{\numitems+1}(u)=-\infty$ if $u<0$. 
\end{definition}

We have the following characterization of the optimal solution, which mirrors the characterization in \autoref{lem:recursive}.  The proof is analogous and is omitted.

\begin{lemma}\label{lem:recursivemultiagent}
	Consider a mechanism $(x,\hat{p})$ defined recursively from $i = 1$ to $k$, given solutions $\hat{g}_i(\vec{c})$ and $(X_i,P_i)({\vec{c}})$ to the cumulative tradeoff problems of \autoref{def:cum tradeoff multi},  as follows.   The mechanism on day 1, $(x_1,\hat{p}_1)$ is equal to $(X_1,P_1)({\vec{c}_0})$, where $\vec{c}_0$ maximizes $\hat{g}_1(\vec{c})$.  The mechanism on day $i \geq 2$ is $x_i(v_{\leq i}) = X_i({\vec{c}_{i-1}},v_i)$, and $\hat{p}_i(v_{\leq i}) = P_i({\vec{c}_{i-1}},v_i)$ where $\vec{c}_{i-1}$ is defined by $c^\agent_{i-1} = U^\agent_{i-1}(\vec{c}_{i-2},v^\agent_{i-1})$ for all $\agent$.  The mechanism $(x,\hat{p})$ is an optimal solution to the adjusted problem~\eqref{eq:adjusted multi agent}.
\end{lemma}

Similar to \autoref{subsec:char}, the preparation above suggests a recursive characterization of the optimal mechanism as follows.

\begin{theorem}\label{thm:characterization multi agent}
	An optimal mechanism for the multi-agent dynamic mechanism design problem is characterized as follows: 
	\begin{enumerate}
		\item[0.] (Pre-processing) Recursively (from $i=k$ down to $1$) define the cumulative tradeoff functions $\hat{g}_i(\vec{c})$,  mechanisms $(X_i({\vec{c}}),P_i({\vec{c}}))$, and corresponding interim utility functions $U_i({\vec{c}})$ for all $i$ and $\vec{c}$ as per (\autoref{def:cum tradeoff multi}).  Set $\vec{c}_0$ equal to the maximizer of $\hat{g}_1(\vec{c})$.
		\item[1.] On each day $i \ge 1$, if agents report $v_i$, then each agent $\kappa$ is allocated with probability $X^\agent_i({\vec{c}_{i-1}},v_i)$, pays $v^\agent_i X^\agent_i({\vec{c}_{i-1}},v_i)$ if $i<k$ or $P^\agent_i({\vec{c}_{i-1}},v_i)$ if $i=k$. For $i\geq 2$ set $c^\agent_{i-1} = U^\agent_{i-1}(\vec{c}_{i-2},v^\agent_{i-1})$.
	\end{enumerate}
\end{theorem}

\section{An Infinite Horizon Problem} \label{sec:infinite horizon}
In this section we consider an infinite horizon version of the single-agent problem where the buyer and the seller discount future utilities with a common discount factor $\delta$.  Throughout the section we assume that the value on each day is drawn i.i.d. from a distribution $F$ with density $f$.  We study the design of simple mechanisms, called ``Markov mechanisms,'' where the allocation $x(v_{i-1},v_i)$ and payment $p(v_{i-1},v_i)$ on each day $i$ only depend on today's report $v_i$ and yesterday's report $v_{i-1}$ (but not the agent's report from $2$ days ago, etc). For the first period, we assume that the seller has access to a realization of the buyer's value $v_0$ drawn from $F$ and it is common knowledge that the seller uses that sample. Alternatively, suppose that the seller has been interacting with the buyer in previous rounds using a single-stage truthful mechanism, and decides to switch to a Markov mechanism. The goal is to maximize the discounted sum of payments:
\begin{align}
&\expect{\sum_{i\geq 1} \delta^i p(v_{i-1},v_i)}.\nonumber
\intertext{Since the value at each stage is drawn independently from an identical distribution, the revenue is}
&\sum_{i \geq 1} \delta^i \expect{p(v_0,v_1)} = {\expect{p(v_0,v_1)}}/{1-\delta}.\nonumber
\intertext{Since $\delta$ is a constant, the seller's problem is to simply maximize}
&\expect[v_0,v_1]{p(v_0,v_1)}, \label{eq:infinite revenue}
\intertext{subject to the periodic incentive compatibility condition.  Recall that the mechanism $(x,p)$ is the same at each stage, and therefore we need to only write the incentive compatibility condition for one stage. That is, for true value $v_1$, misreport $v'_1$, and the value on the day before $v_0$, we must have}
&v_1x(v_0,v_1) - p(v_0,v_1) + \delta \expect[v_2]{v_2x(v_1,v_2) - p(v_1,v_2)}+ U_3\nonumber \\
&\geq  v_1x(v_0,v'_1) - p(v_0,v'_1) + \delta \expect[v_2]{v_2x(v'_1,v_2) - p(v'_1,v_2)} + U_3,\nonumber
\intertext{where $U_3$ is the appropriately discounted expected utility of the buyer from stages $3$ and onwards, which importantly is not affected by the potential misreport $v'_1$.  Hence, the incentive compatibility constraint becomes}
&v_1x(v_0,v_1) - p(v_0,v_1) + \delta \expect[v_2]{v_2x(v_1,v_2) - p(v_1,v_2)} \\
&\geq  v_1x(v_0,v'_1) - p(v_0,v'_1) + \delta \expect[v_2]{v_2x(v'_1,v_2) - p(v'_1,v_2)}.\label{eq:infinite PIC}
\intertext{Finally, we require the following ex-post individual rationality condition for each $v_1$ and $v_0$,}
&v_1x(v_0,v_1) - p(v_0,v_1)  \geq 0.\label{eq:infinite IR}
\end{align}

Our analysis resembles the analysis of \autoref{sec:dynamic mechanism design and optimal revenue-utility tradeoff} closely.  In particular, define the \emph{adjusted payment} $\hat{p}(v_0,v_1) = p(v_0,v_1) - U(v_1)$, where $U(v_1) = \delta \expect[v_2]{v_2x(v_1,v_2) - p(v_1,v_2)}$.  The problem becomes
\begin{align*}
\max_{x,p} &\expect[v_0,v_1]{\hat{p}(v_0,v_1) + U(v_1)}\\
\text{s.t., } & v_1x(v_0,v_1) - \hat{p}(v_0,v_1) \geq  v_1x(v_0,v'_1) - \hat{p}(v_0,v'_1) \\
& v_1x(v_0,v_1) - \hat{p}(v_0,v_1)  \geq U(v_1).
\end{align*}

Define $\hat{p}(v) = \expect[v_0]{\hat{p}(v_0,v)}$, $x(v) = \expect[v_0]{x(v_0,v)}$, $U = \expect[v]{U(v)}$, and relax the problem by taking the expectation of the PIC and ex-post IR constraints.
\begin{align*}
\max_{x,p,U} &\expect[v]{p(v)}+U\\
\text{s.t., } & vx(v) - \hat{p}(v) \geq  vx(v') - \hat{p}(v') \\
& vx(v) - \hat{p}(v)  \geq U.
\end{align*}

Consider any feasible solution $(x,p)$ to the relaxed problem with some $U=\tilde{U}$.  Note that the alternative mechanism that adds $\tilde{U}$ to all payments is a feasible with $U=0$, and its objective value is the same, namely $\expect[v]{p(v)}+\tilde{U}$.  As a result, without loss of generality we can assume that $U=0$.  Note that with this simplification the problem reduces to a static monopoly problem, where the solution is to simply post the monopoly price. In turn, a static solution is feasible for the un-relaxed problem. We conclude with the main theorem of this section that Markov mechanism extract the same revenue as single-stage mechanisms.

\begin{theorem}
Consider an infinite horizon problem with i.i.d. bidder values and discount factor $\delta$, where each day's mechanism is the same and allowed to only depend on the current day's report and the previous day's report (but not on any other day's report).  The revenue maximizing such mechanism subject to periodic inventive compatibility \eqref{eq:infinite PIC} and ex-post individual rationality \eqref{eq:infinite IR} is to simply post the monopoly price on each day, independent of the history.
\end{theorem}

\bibliographystyle{acmsmall}
\bibliography{bibs}

\appendix
\section{The Utility-constrained Surplus Optimization Problem}\label{sec:utility constrained surplus}
In this section we study an important special case of the revenue-utility tradeoff problem, which we term the  \emph{utility-constrained surplus maximization} problem.  Studying this problem is important for two reasons.  First, we will use the analysis of this problem to further simplify the structure of the optimal solution for the case of $2$ days.  Second, the solution to this problem will be used later in \autoref{subsec:2approx} to identify a simple 2-approximation mechanism with $k$ days.

\begin{definition}[Utility-Constrained Surplus Function]\label{def:ucsurplus}
For $i \in [1: \numitems]$ define the \emph{utility-constrained surplus function} ${g}_{i}(c)$ to be the value of the solution to the revenue utility tradeoff problem of Definition~\ref{def:revenue utility tradeoff general} for distribution $f_i$, utility bound $c$, and tradeoff function $g(x) = x$ if $x\geq 0$, and $g(x) = -\infty$ otherwise.
\end{definition}

Note that $g(x) = -\infty$ for $x<0$ is essentially adding an extra constraint to the problem that the utility must be non-negative, that is, the mechanism must be individually rational.  Recall that in the revenue utility tradeoff problem the objective is to maximize $\expect{p(v) + g(u(v))}$.  Given that $g(x) = x$, the objective simplifies to $\expect{p(v) + u(v)} = \expect{p(v) + vx(v) - p(v)} = \expect{vx(v)}$, which is simply to optimize the expected surplus (hence the choice of name).  In particular, the problem is

\begin{align*}
&\max_{x,p} \expect[v\sim f_i]{vx(v)} \\
&\text{s.t. IC: } vx(v) - p(v) \geq vx(v') - p(v'), \forall v,v' \\
&\text{IR: } vx(v) - p(v) \geq 0,  \\
& \expect[v \sim f_i]{vx(v) - p(v)} \leq c.
\end{align*}

In the rest of this subsection we will characterize the structure of the solution to utility-constrained surplus optimization.  We will use the solution in the future subsections to obtain simple approximation mechanisms to the dynamic problem.  Before we proceed, compare the definition of the cumulative tradeoff function $\hat{g}_i$ of \autoref{def:cum tradeoff} to the utility-constrained surplus function $g_i$ of \autoref{def:ucsurplus}.  Whereas the cumulative tradeoff function $\hat{g}_i$ is defined recursively, function $g_i$ is simply defined separately for each day $i$.  In addition, as we will see next, the optimal solution to the utility-constrained surplus problem $g_i$ has a very simple form, whereas the solution to the $\hat{g}_i$ problem can be very complicated (we will further study the general revenue-utility tradeoff problem in \autoref{sec:softproblem}).  The simple structure of utility-constrained surplus problem will be useful in the characterization of approximate mechanisms for the dynamic problem.

We now characterize the structure of the optimal utility-constrained surplus maximization.  Fix a distribution $f_i$ (and drop the index $i$ for simplicity).  For $p \in [\underline{v},\bar{v}]$, let $u(p)$ and $S(p)$ be the expected utility and surplus of posting a price $p$, respectively. That is, $u(p) = \int_{v \geq p} (v-p) f(v) dv$ and $S(p) = \int_{v\geq p} v f(v) dv$.  For $c \in [0,\expect{v}]$, let $p(c)$ be the price that gives the buyer expected utility $c$, that is, $p(c) = u^{-1}(c)$ (notice that $p(c)$ exists since the expected utility of posting a price $p$ is continuous and strictly decreasing in $p$).  Let ${S}_U(c)$ be the expected surplus of posting the price $p(c)$, that is, $S_U(c) := S(p(c))$.

\begin{lemma}\label{lem:gstructure}
If $c \geq \expect[v]{v}$, any solution that gives the item to all types and charges a constant in $[0,-c+\expect[v]{v}]$ is optimal.  Otherwise, the solution is a randomization over two prices $p_1\leq p_2$, that is, $x(v) = 0$ for all $v\leq p_1$, $x(v) =\alpha$ for all $p_1 \leq v \leq p_2$, and $x(v) =1$ otherwise. If $S_U$ is concave, the solution is simply to post a deterministic price, that is, $p_1=p_2$.  In either case, there exists an optimal solution that satisfies the upper bound $c$ on expected utility with equality.  Additionally, the utility-constrained surplus maximization function $g(c)$ is a concave function of $c$. 
\end{lemma}
\begin{proof}
Feasibility requires that $x(v)\leq 1$, and as a result, the value of the problem $\expect[v]{vx(v)}$ is at most $\expect[v]{v}$.  When $c \geq E_v[v]$, this upper bound can be achieved by assigning the item deterministically, and paying all types a constant in $[0,c-\expect[v]{v}]$.  In particular, by giving all types a payment $c-\expect[v]{v}]$, the expected utility will be equal to $c$.  Note that the value of the problem is monotone non-decreasing in $c$ since the problem becomes more relaxed as $c$ increases.  Thus, if $c \leq \expect[v]{v}$, we can decrease the payment of all types and satisfy $\expect[v]{u(v)}=c$. As a result, in either case there exists an optimal solution such that $\expect[v]{u(v)} =c$.

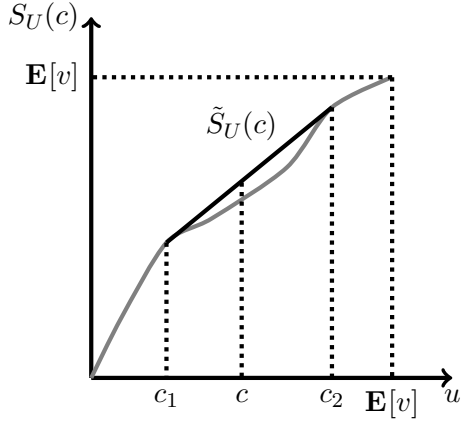
\begin{figure}
\centering
    \begin{tikzpicture}[domain=0:3, scale=4, ultra thick, align=center]    
    \draw[->] (0,0) -- (1.2,0) node[below]{$u$};
    \draw[->] (0,0) -- (0,1.2) node[left]{$S_U(c)$};
       \draw[black!50] plot [smooth] coordinates {(0,0) (.1,.2) (.25,.45) (.4,.53) (.65,.7) (.8,.9) (1,1)};

    \draw[dotted] (1,1) -- (1,0) node[below]{$\expect{v}$};
    \draw[dotted] (1,1) -- (0,1) node[left]{$\expect{v}$};

    
    \draw (.25,.45) -- (.8,.9);
    \draw (.5,.75) node[above]{$\tilde{S}_U(c)$};
    
    \draw[dotted] (.25,.45) -- (.25,0) node[below]{$c_1$};
    \draw[dotted] (.8,.9) -- (.8,0) node[below]{$c_2$}; 
    \draw[dotted] (.5,.65) -- (.5,0) node[below]{$c$};
\end{tikzpicture}
\caption{The surplus of posting a price $p(u)$ is monotone non-decreasing in $u$.  As $u$ increases from $0$ to $\expect{v}$, $S_U$ increases from $0$ to $\expect{v}$ (where the item will be offered for a price of $0$).  At a point $c$ where $S_U(c) < \tilde{S}_U(c)$, the optimal surplus is achieved by optimizing over prices $p(c_1)$ and $p(c_2)$.}
\label{fig:utility constrained suprlus}
\end{figure}

We now prove the second part of the lemma. Recall that an incentive compatible mechanism is a convex combination of posted prices (\autoref{lem:randomization over prices}).  Since $u \geq 0$, the optimal solution must satisfy $p(0)\leq 0$.  In the optimum solution, $p(0) =0$ as otherwise if $p(0)<0$ we can shift all prices up while respecting the IR condition and relaxing the upper bound on the expected utility.  As a result, in the utility-constrained surplus maximization problem the objective is to calculate a distribution over prices denoted with cumulative density function $G(p)$ such that the expected surplus is optimized, subject to the constraint the the expected utility equals $c$, that is,
\begin{align*}
&\max_{G} \int_{p \geq \underline{v}}^{\bar{v}} S(p)dG(p) \\
&\text{s.t.}  \int_{p \geq \underline{v}}^{\bar{v}} u(p)dG(p) = c.
\end{align*}

We can represent the above problem using a distribution over target utilities.  To see this, for a given distribution $G$ on prices, consider a distribution $G_U$ induced by first drawing a price $p$ from $G$, and then mapping that price to $u(p)$.  An identical mechanism draws $u$ from $G_U$, and posts the price $p(u)$.  Thus we can rewrite the above problem as follows
\begin{align*}
&\max_{G_U} \int_{c \geq 0}^{\expect{v}} S_U(c)dG_U(c) \\
&\text{s.t.}  \int_{c \geq 0}^{\expect{v}} u(c)dG_U(c) = c.
\end{align*}

Note that the above program is the definition of the concave hull $\tilde{S}_U$ of $S_U$.   The surplus $\tilde{S}_U(c)$ is achievable by randomizing over at most two prices $p_1 = p(u_1)$ and $p_2 = p(u_2)$ (see \autoref{fig:utility constrained suprlus}).  If $S_U$ is concave, the optimum surplus is achieved by simply posing the price $p(c)$.

The analysis above also proves concavity of $g$.
\end{proof}

Having revealed the simple structure of the solution to the utility-constrained surplus maximization problem, we will now use it to simplify the structure of the optimal mechanism with $2$ days, as identified by \autoref{lem:recursive}.  Recall the recursive definition of tradeoff functions $\hat{g}_i$ from \autoref{sec:dynamic mechanism design and optimal revenue-utility tradeoff} (\autoref{def:cum tradeoff}).  Notice that $\hat{g}_k = g_k$.  In particular, when $k=2$, $\hat{g}_2 = g_2$. Thus in the optimal solution for $k=2$, the mechanism $(x_1,\hat{p}_1)$ at day $1$ optimizes

\begin{align*}
\expect[v_1]{\hat{p}_1(v_1) + g_2(u(v_1))},
\end{align*} 

\noindent and the mechanism at day $2$ is simply the solution to $g(v_1x_1(v_1) - \hat{p}_1(v_1))$, which by lemma \autoref{lem:gstructure} is a randomization over two prices that give buyer the expected utility $v_1x_1(v_1) - \hat{p}_1(v_1)$.

\begin{example}[The Equal Revenue Distribution at Stage 1, continued]\label{ex:equal revenue day 1 contd}
Consider a 2-stage problem were the first distribution is an equal revenue distribution.  Recall from \autoref{ex:equal revenue day 1} that $(x_1,\hat{p}_1) = (1,1)$, and the mechanism at day two is the solution to utility-constrained surplus optimization problem with utility upper bound $v_1-1$, and lower bound $0$.  \autoref{lem:gstructure} shows that the mechanism at day two is simply a randomization over two prices that give the buyer expected utility $v_1-1$.  This analysis, together with \autoref{lem:changeofvariables} that maps the solution to the adjusted problem to the solution to the original problem, suggests the following mechanism is optimal to the original problem.
\begin{enumerate}
\item Allocate the item at stage 1; charge $v_1$.
\item At stage 2, randomize over two prices that give buyer expected utility $v_1$.
\end{enumerate} 
\end{example}

In the next subsection we will need to study a variant of the utility-constrained surplus problem where in addition to an upper bound $c$ on expected utility, a pointwise lower bound $c_L$ on utility is also required, that is $u(v) \geq c_L$ for all $v$.  The value of the optimum solution to this problem is equal to $g_i(c-c_L)$, that is, the optimum individually rational solution with an upper bound $c-c_L$ on expected utility.  The reason is that any feasible solution with pointwise lower bound $c_L$ on utility can be converted to a feasible solution to $g(c-c_L)$, with the same objective value, by adding $c_L$ to all payments.  Similarly, given an individually rational mechanism, we can subtract from all payments and get a mechanism with pointwise utility lower bound of $c_L$.  Invoking \autoref{lem:gstructure}, we have the following lemma.

\begin{lemma}
Consider the utility-constrained surplus optimization problem with additional lower pointwise lower bound on utility, $u(v) \geq c_L$ for all $v$.  The solution is to offer a randomization over two prices for the item obtained by solving $g_i(c-c_L)$, in addition to a payment $-c_L$ for all types.
\end{lemma}\label{lem:utility shift}

\section{The \softproblem Problem}\label{sec:softproblem}
This section discusses a partial characterization of the solution to a generalization of the revenue-utility tradeoff problem.  

The revenue-utility tradeoff problem is stated in a much simpler manner in terms of the utility function of a mechanism, instead of the more standard way of expressing the problem in terms of the allocation function.  As such, we will use an expression of revenue directly in terms of the utility function, in contrast to Myerson's formulation in terms of the allocation function \autoref{lem:myerson's lemma}.  The representation of revenue in terms of the utility function will allow us to perform point-wise comparisons that would have not been possible with the more common representation of revenue in terms of the allocation function. 
 The following analysis is standard and is included for completeness.  Define the \emph{revenue function} $R(v):= v(1-F(v))$  to be the expected revenue from posting a price $v$ for the item.

\begin{lemma}\label{lem:revenue by utility function}
The revenue of an incentive compatible mechanism with utility function $u$ is
\begin{align*}
 u(v)(vf(v))|_{v=\underline{v}}^{\bar{v}} + E_v\left[u(v)\frac{R''(v)}{f(v)}\right].
 \end{align*}
\end{lemma}
\begin{proof}
We express revenue using the equation $p(v) = vu'(v) - u(v)$ (\autoref{lem:myerson's lemma}) to get
\begin{align*}
E_v[p(v)]  = E_v[v\cdot u'(v) - u(v)]  &= \int_v vu'(v)f(v)dv - \int_v u(v)f(v)dv \\
&= u(v)(vf(v))|_{v=\underline{v}}^{\bar{v}} - \int_v u(v)(vf(v))' dv - \int_v u(v)f(v)dv \\
&= u(v)(vf(v))|_{v=\underline{v}}^{\bar{v}} - \int_v u(v)[(vf(v))'+f(v)] dv,
\intertext{where the third equation followed from an integration by parts, and the forth equation by rearranging terms.   Notice that $R'(v) = 1-F(v) - vf(v)$, and $R''(v) = -f(v) - (vf(v))'$.  As a result, the revenue is}
E_v[p(v)] &= u(v)(vf(v))|_{v=\underline{v}}^{\bar{v}} + \int_v u(v)R''(v) dv\\
&= u(v)(vf(v))|_{v=\underline{v}}^{\bar{v}} + E_v\left[u(v)\frac{R''(v)}{f(v)}\right].
\end{align*}
\end{proof}

Notice that in the special case where $\underline{v}=0$, $u(v)(vf(v))|_{v=\underline{v}}^{\bar{v}}$ simplifies to $u(\bar{v})\bar{v}f(\bar{v})$.

\begin{definition}\label{def:revenue utility tradeoff type dependent}
The \emph{revenue-utility-tradeoff} with \emph{type-dependent} tradeoff is parameterized by a single-dimensional distribution $f$, a utility bound $c \in \reals$, and a tradeoff function $g_{dsw}: [\underline{v},\bar{v}] \times \reals \rightarrow \reals$ and is defined as follows
\begin{align*}
\max_{x,p}  &\expect[v\sim f]{p(v)+ g_{dsw}(v,vx(v)-p(v))}\\
\text{s.t., IC: } & vx(v) - p(v) \geq vx(v') - p(v')
\end{align*}
\end{definition}

Note that the above formulation is more general than the case in \autoref{sec:dynamic mechanism design and optimal revenue-utility tradeoff} where the tradeoff function was only a function of a single parameter $c$.  By generalizing the tradeoff function to be also a function of the type, the problem includes for example revenue optimization problems where the individual rationality constraint \emph{depends on the type}.  That is, consider a revenue maximization problem where the utility of each type must be lower bounded by a given function $\ell(v)$.  By setting $g_{dsw}(v,c) = 0$ if $c \geq \ell(v)$, and $g_{dsw} = -\infty$ otherwise, we can capture this problem as an instance of the type-dependent revenue utility tradeoff problem.

The analysis of \autoref{sec:prelim} (\autoref{lem:revenue by utility function}) allows us to express revenue in terms of the utility function, and thus \tdsoftproblem problem can be expressed as follows:

\begin{align*}
&\max_{u}  E_v\Big[u'(v)\phi(v) + g_{dsw}(v,u(v))\Big] - u(0) = E_v\Big[u(v)\frac{R''(v)}{f(v)} + g_{dsw}(v,u(v))\Big] + u(v)(vf(v))|_{v=\bar{v}} \\
&\text{s.t. } 0\leq u'(v) \leq 1, 0\leq u''(v).
\end{align*}

If for all $v$, $g_{dsw}(v,u)$ is a monotone non-decreasing and concave function of $u$, that is, $\partial_2 g_{dsw}(v,u)\geq 0$ and $\partial^2_2g_{dsw}(v,u) \leq 0$ ($\partial_i g$ is the partial derivative of $g$ with respect to its $i$'th variable), then the problem admits a very simple 2-approximation: $u'(v) = x(v) = 1/2$, for all $v$ that are less than the {monopoly reserve corresponding to $f$}, and $u'(v)= x(v)=1$ otherwise. This allocation can be seen as a randomization over two allocations; one is the allocation induced by posting the monopoly reserve, and the other is the constant allocation $x(v) = 1$. The 2-approximation guarantee simply follows from the fact that the first allocation maximizes expected revenue (c.f.~\cite{M81}), while the second  maximizes $E_v[g_{dsw}(v,u(v))]$, together with the concavity of $g$ and the linearity of expectation. 

\begin{theorem}\label{thm:tdsoftproblemapprox}
Consider an instance of the \tdsoftproblem problem of Definition~\autoref{def:revenue utility tradeoff general}. If $\partial_2 g_{dsw}(v,u)\geq 0$ and $\partial^2_2g_{dsw}(v,u) \leq 0$, the allocation that satisfies $x(v)=1/2$, for all $v$ less than the monopoly reserve corresponding to the type distribution, and $x(v)=1$, otherwise, defines a mechanism that is a $2$-approximation to the optimal solution.
\end{theorem}
Recall from Lemma~\ref{lem:gstructure} that the conditions of the theorem are satisfied by the utility-constrained surplus optimization function $g(c) = g_{dsw}(v,u)$. 

Next we study the optimality of a natural generalization of the allocation rule considered above. A \emph{$(\alpha,\nu)$-step allocation} satisfies $x(v) = \alpha$ for all $v \leq \nu$, and $x(v) = 1$ otherwise. Intuitively such an allocation attempts to optimize both revenue and utility in parallel by randomizing between a pricing mechanism and an allocation that is equal to 1 everywhere.  The following theorem specifies conditions under which such a mechanism is optimal.

\begin{theorem}
Consider an instance of the \tdsoftproblem problem of Definition~\ref{def:revenue utility tradeoff general}. The optimal allocation is an $(\alpha,\nu)$-step allocation if $\frac{d}{dv}(\frac{R''(v)}{f(v)}) \leq 0$, $\partial^2_2 g_{dsw}(v,u)\leq 0$, and $\partial_1\partial_2 g_{dsw}(v,u)\leq 0$, where $R(v)=v \cdot (1-F(v))$.
\label{thm:tdsoftproblem}\end{theorem}

\begin{proof}
Fix the value of the utility at $\bar{v}$ and consider a feasible solution $u(v)$ taking that value at $\bar{v}$. The infinitesimal change in the objective value that would result from an infinitesimal change in the utility $u(v)$ of a certain type $v$ is
\begin{align*}
\left(\frac{R''(v)}{f(v)} + \partial_2 g_{dsw}(v,u(v))\right) \times f(v) du(v) dv.
\end{align*}

\begin{figure}[h!]
\centering
    \begin{tikzpicture}[domain=0:3, scale=3.4, ultra thick]    
    \draw[->] (0,0) -- (1.1,0);
    \draw[->] (0,0) -- (0,1.1);
    \draw (1,.05) -- (1,-.05) node[below]{$\bar{v}$}; 
    \draw[dashed] (1,.8) -- (1,0);
    \draw (1,.8) node[above]{$u$}-- (.5,.3) -- (0,.2);
    \draw[dotted] (0.1,.23) -- (.8,.6);
    \draw[->] (.5,.3) -- (.43,.4);
    \draw (.5,-0.1) node[below]{(a)};
    
    \draw[->] (1.5,0) -- (2.6,0);
    \draw[->] (1.5,0) -- (1.5,1.1);
    \draw (2.5,.05) -- (2.5,-.05) node[below]{$\bar{v}$}; 
    \draw[dashed] (2.5,.8) -- (2.5,0);
    \draw (2.5,.8) node[above]{$u$}-- (2,.3) -- (1.5,.2);
    \draw[dotted] (2.5,.8) -- (2.1,.25) -- (1.5,.2);
    \draw[->] (2,.3) -- (2.1,.25);
    \draw (2,-0.1) node[below]{(b)};

    \draw[->] (3,0) -- (4.1,0);
    \draw[->] (3,0) -- (3,1.1);
    \draw (4,.05) -- (4,-.05) node[below]{$\bar{v}$}; 
    \draw[dashed] (4,.8) -- (4,0);
    \draw (4,.8) node[above]{$u$}-- (3.5,.3) -- (3,.2);
    \draw[dashed] (3.7,.5) -- (3.7,0);
    \draw (3.7,.05) -- (3.7,-0.05) node[below]{$v$};
    \draw[dotted] (3,.2) -- (3.85,.55) -- (4,.8);
    \draw (3.5,.4) node[above]{$\hat{u}$};
    \draw (3.5,-0.1) node[below]{(c)};
\end{tikzpicture}
\caption{(a) In this case, the marginal value of increasing $u$ locally is positive everywhere. As a result, unless $u$ is a straight line connecting $(0,u(0))$ to $(\bar{v},u(\bar{v})$, the objective value of $u$ can be improved by a local increase. (b) In this case, the marginal value of increasing $u$ locally is negative everywhere. As a result, unless $u$ is as low as possible given the choices of $u(0)$ and $u(\bar{v})$, the objective value of $u$ can be improved by a local decrease. (c) The function $\hat{u}$ is constructed from $u$ by taking the maximum of a line that connects $(0,u(0))$ to $(v,u(v))$, and a 45 degree line passing through $(\bar{v},u(\bar{v}))$. The function $\hat{u}$ is feasible if $u$ is feasible, is an upper bound on $u$ below $v$ and a lower bound on $u$ above $v$. So if $u$ is optimal then so is $\hat{u}$ (the functions can be both optimal if the marginal value of changing $u$ is zero).}
\label{figtdsoft}
\end{figure}
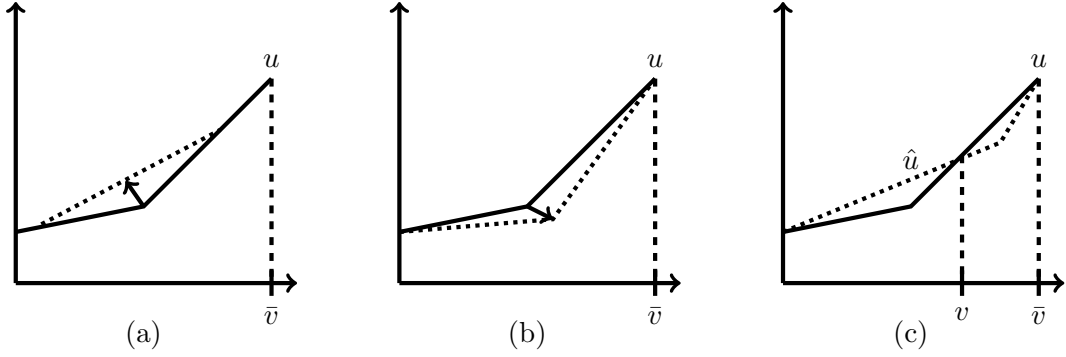

First, consider the case where $\frac{R''(v)}{f(v)} + \partial_2 g_{dsw}(v,u(v))>0$, for all $v$. Then $u$ can be an optimal solution only if $u'(v)$ is a constant everywhere, since otherwise we can increase $u$ locally and increase the objective value, while respecting the feasibility constraints (see \autoref{figtdsoft}, case (a)). This solution corresponds to the special case where $\nu = \bar{v}$.

Similarly, consider the case where $\frac{R''(v)}{f(v)} + \partial_2 g_{dsw}(v,u(v))<0$ for all $v$.  Then $u$ can be optimal only if the feasibility condition does not allow for lowering the utility function anywhere, corresponding to the case that $u(v) = 0$ for $v \leq \bar{v} - u(\bar{v})$ and $u(v) = v-(\bar{v} - u(\bar{v}))$ for $v \geq \bar{v} - u(\bar{v})$ (see \autoref{figtdsoft}, case (b)).

The most interesting case is when $\frac{R''(v)}{f(v)} + \partial_2 g_{dsw}(v,u(v)) = 0$, for some $v$. The assumptions of the theorem imply that $\frac{R''(v')}{f(v')} + \partial_2 g_{dsw}(v',u') \geq 0$ for all $v' \leq v$ and $u' \leq u(v)$, and that $\frac{R''(v'')}{f(v'')} + \partial_2 g_{dsw}(v'',u'') \leq 0$ for all $v'' \geq v$ and $u'' \geq u(v)$. As a result, if a feasible $\hat{u}$ exists that upper bounds $u(v')$ for all $v'\leq {v}$ and lower bounds $u(v'')$ for all $v'' \geq v$, then $\hat{u}$ will have an objective value no less than objective value of $u$. Consider a utility function $\hat{u}$ that passes through points $(0,u(0))$, $(v,u(v))$ and $(\bar{v},u(\bar{v}))$ such that $\hat{u}'(v) = \alpha <1$, for all $v$ less than some threshold $\tilde{v}$, and $\hat{u}'(v) = 1$, for all $v\geq \tilde{v}$ (see \autoref{figtdsoft}, case (c)). By convexity of $u$ and the fact that $u'\leq 1$, we must have $\hat{u}(v')\geq u(v')$, for all $v'\leq v$, and $\hat{u}(v'') \geq u(v'')$, for all $v'' \geq v$. The objective value of $\hat{u}$ must therefore be no less than that of $u$. Notice that the allocation corresponding to $\hat{u}$ is an $(\alpha, \nu)$ allocation.
\end{proof}

Since $R''(v) = -2f(v) - vf'(v)$, the condition $\frac{d}{dv}(\frac{R''(v)}{f(v)}) \leq 0$ can be alternatively expressed as
\begin{align*}
\frac{d}{dv}(\frac{vf'(v)}{f(v)})\geq 0.
\end{align*}

Notice that the \tisoftproblem is a special case of the \tdsoftproblem where $\partial_1 g_{dsw}(v,u) = 0$, so that one of the conditions of \autoref{thm:tdsoftproblem} is automatically satisfied.
\begin{corollary}
Consider an instance of the \tisoftproblem problem of Definition~\ref{def:revenue utility tradeoff general}. The optimal allocation is an $(\alpha,\nu)$-step allocation if $\frac{d}{dv}(\frac{R''(v)}{f(v)}) \leq 0$ and $g_{isw}$ is concave,  where $R(v)=v \cdot (1-F(v))$.
\end{corollary}

\section{FPTAS}\label{app:FPTAS}
This section sketches how to compute the cumulative tradeoff functions $\hat{g}_i(c)$ and the corresponding optimal mechanisms $(X^c_i,P^c_i)$ of \autoref{def:cum tradeoff} and \autoref{def:revenue utility tradeoff general} efficiently. In turn, efficient computation of these functions is required in order to efficiently compute the optimal mechanisms in \autoref{thm:characterization}.

Recall the recursive definition of cumulative tradeoff functions through the following programs:
\begin{align}
\hat{g}_i(c) = \max_{x,p}  &\expect[v\sim f_i]{vx(v)+ \hat{g}_{i+1}(vx(v)-p(v))}\label{eq:problem1}\\
\text{s.t., IC: } & vx(v) - p(v) \geq vx(v') - p(v')\nonumber \\
& c = \expect[v \sim f_i]{vx(v)-p(v)}.\nonumber
\end{align}
The resulting functions are concave as can be established via an inductive argument similar to that in Lemma~\ref{lem:aux1}. Computing function $\hat{g}_i(\cdot)$ given $\hat{g}_{i+1}(\cdot)$ \emph{exactly} would require access to the complete description of function $\hat{g}_{i+1}(\cdot)$, which is computationally infeasible since the function is defined over a continuous domain.  

We will argue instead that there exist concave functions $\tilde{g}_i(\cdot)$ that approximate functions $\hat{g}_i(\cdot)$, and are piecewise linear with polynomially many pieces.  In what follows we describe how these piecewise linear functions $\tilde{g}_i(\cdot)$ are defined, and how the error propagates as we recursively define them from $i=k$ down to $i=1$.  To do this, assume that a function $\tilde{g}_{i+1}(\cdot)$ is given to us, which approximates function $\hat{g}_{i+1}(\cdot)$ within an additive error $\delta_{i+1}$ (to be set later), that is, $\tilde{g}_{i+1}(c) \in [\hat{g}_{i+1}(c)-\delta_{i+1},\hat{g}_{i+1}(c)]$ for all $c$, and also assume that $\tilde{g}_{i+1}(\cdot)$ consists of polynomially many pieces.  Now define $\bar{g}_i(\cdot)$ as the solution to the surplus-utility tradeoff problem with tradeoff function $\tilde{g}_{i+1}(\cdot)$, that is,
\begin{align}
\bar{g}_i(c) = \max_{x,p}  &\expect[v\sim f_i]{vx(v)+ \tilde{g}_{i+1}(vx(v)-p(v))}\label{eq:problem2}\\
\text{s.t., IC: } & vx(v) - p(v) \geq vx(v') - p(v') \nonumber \\
& c = \expect[v \sim f_i]{vx(v)-p(v)}.\nonumber
\end{align}
%
Note that if $\tilde{g}_{i+1}$ is concave, then the afore-described program is convex, and results in a convex function $\bar{g}_i(\cdot)$. Additionally, given the assumption that $\tilde{g}_{i+1}$ approximates $\hat{g}_{i+1}$ within a distance $\delta_{i+1}$, the problem~\eqref{eq:problem2} is approximately equal to the problem~\eqref{eq:problem1}, and thus the function $\bar{g}_i$ must also be within a distance $\delta_{i+1}$, that is, $\bar{g}_i(c) \in [\hat{g}_i(c)-\delta_{i+1}, \hat{g}_i(c)]$ for all $c$.  Even though $\bar{g}_i$ is a good approximation to $\hat{g}_i$, it is not necessarily easy to describe succinctly.  We therefore define a piecewise linear function $\tilde{g}_i$ given $\bar{g}_i$ as follows.  Fix a parameter $\delta'_i$ (to be set later).  We define $\tilde{g}_i$ to coincide with $\bar{g}_i$ when the value of $\bar{g}_i(c)$ is a multiple of $\delta'_i$, and extend $\tilde{g}_i$ to be linear in between.  In particular, $\tilde{g}_i(c) = \bar{g}_i(c)$ if $\bar{g}_i(c) = m\delta'_i$ for some integer $m$.  Note that this construction ensures that $\tilde{g}_i$ is within a distance $\delta'_i$ of $\bar{g}_i$, and thus within a distance $\delta'_i + \delta_{i+1}$ from $\hat{g}_i$.  We next argue that $\tilde{g}_i$ needs polynomially many pieces in $1/\delta'_i$ to be described.  Note that the value of the function $\hat{g}_i(\cdot)$ is non-negative and is at most $k\bar{v}$, where $k$ is the number of days and $\bar{v}$ is the maximum possible value over all days.  As a result, and by concavity of $\hat{g}_i$, we need at most $2k\bar{v}/\delta'_i$ pieces to define $\tilde{g}_{i}$. Finally, $\tilde{g}_{i}$ is concave by construction. To complete our recursive definition of the $\tilde{g}_{i}$'s we define $\tilde{g}_{k+1} = \hat{g}_{k+1}$, which is fine since $\hat{g}_{k+1}$ only has two pieces.

Given the above analysis, we can set parameters $\delta'_i$ appropriately such that functions $\tilde{g}_i(\cdot)$ are within a distance of $\epsilon$ from functions $\hat{g}_i(\cdot)$ for all $i$, and given any desired error $\epsilon$.  Since the error propagates linearly, we simply need to set $\delta'_i = \epsilon/k$ for all $i$.  Hence, throughout in the afore-described construction the functions $\tilde{g}_i$ have polynomially in $1/\epsilon$, $k$ and $\bar{v}$ pieces.

The afore-described construction works for discrete and continuous value distributions $f_i$ alike. In the case where  these distributions have discrete support, we observe that all intermediate problems~\eqref{eq:problem2} are finite-dimensional convex programs. So we can adapt the above construction to obtain the following theorem.

\begin{theorem}\label{thm:fptas}
	For any desired error $\epsilon >0$, a mechanism whose revenue is within an additive error $\epsilon$ from that of the optimal mechanism of \autoref{thm:characterization} can be computed in time polynomial in the size of the support of each $f_i$, the maximum value in the support of each $f_i$, the number of days $k$, and $1/\epsilon$.
\end{theorem}

\end{document}